\documentclass[aps,prb,reprint,onecolumn,notitlepage,superscriptaddress, nofootinbib]{revtex4-2}
\usepackage{graphicx}
\usepackage{amsmath,amssymb}
\usepackage{color}
\usepackage{bm}
\usepackage{bbm}
\usepackage{mathtools}
\usepackage{hyperref}
\usepackage{amsthm}
\usepackage{ascmac}
\usepackage[dvipsnames]{xcolor}
\hypersetup{
    colorlinks=true,
    linkcolor=black,
    filecolor=red,      
    urlcolor=blue,
    citecolor=purple
    }
\usepackage{braket}

\newcommand{\Z}{\mathbb{Z}}

\begin{document}

\title{Projective Representations, Bogomolov Multiplier, and Their Applications in Physics}

\author{Ryohei Kobayashi}\email{ryok@ias.edu}
\affiliation{School of Natural Sciences, Institute for Advanced Study, Princeton, NJ 08540, USA}
\affiliation{Theoretical Sciences Visiting Program, Okinawa Institute of Science and Technology Graduate University, Onna, 904-0495, Japan}

\author{Haruki Watanabe}\email{hwatanabe@g.ecc.u-tokyo.ac.jp}
\affiliation{Department of Applied Physics, The University of Tokyo, Tokyo 113-8656, Japan}
\affiliation{Theoretical Sciences Visiting Program, Okinawa Institute of Science and Technology Graduate University, Onna, 904-0495, Japan}

\date{\today}

\begin{abstract}
We present a pedagogical review of projective representations of finite groups and their physical applications in quantum many-body systems. Some of our physical results are new.
We begin with a self-contained introduction to projective representations, highlighting the role of group cohomology, representation theory, and classification of irreducible projective representations.
We then focus on a special subset of cohomology classes, known as the Bogomolov multiplier, which consists of cocycles that are symmetric on commuting pairs but remain nontrivial in group cohomology. Such cocycles have important physical implications: they characterize (1+1)D SPT phases that cannot be detected by string order parameters
and give rise, upon gauging, to distinct gapped phases with completely broken non-invertible $\mathrm{Rep}(G)$ symmetry. We construct explicit lattice models for these phases and demonstrate how they are distinguished by the fusion rules of local order parameters. We show that a pair of completely broken $\mathrm{Rep}(G)$ SSB phases host nontrivial \textit{interface modes} at their domain walls. As an example, we construct a lattice model where the ground state degeneracy on a ring increases from 32 without interfaces to 56 with interfaces. 
\end{abstract}

\newtheorem{example}{Example}
\newtheorem{theorem}{Theorem}
\newtheorem{lemma}{Lemma}
\maketitle
\tableofcontents

\section{Introduction}
One of the most fundamental tools for describing symmetries in physics is the representation theory of groups.
In quantum mechanics, in particular, the transformation rules for quantum states often take the form of projective representations.
For instance, an electron possesses an internal spin-$1/2$ degree of freedom, corresponding to a two-dimensional representation of the $\mathrm{SU}(2)$ symmetry; however, its faithful symmetry action must be viewed as a projective representation of $\mathrm{SO}(3)$, which is realized, for example, through spatial rotation symmetry.
Formally, a projective representation is defined as a map $U\colon G\to U(d)$ satisfying
\begin{align}
U(g)U(g')=\alpha(g,g')U(gg'),
\end{align}
for all $g,g'\in G$, where $\alpha(g,g')\in U(1)$ is a phase factor. 

This mathematical structure is far from merely abstract; it plays an essential role in various branches of modern physics.
In one spatial dimension, for example, symmetry-protected topological (SPT) phases are classified by the projective representations of their internal symmetry group \cite{PhysRevB.83.035107,PhysRevB.85.075125}.
In crystalline solids, electronic states are represented by Bloch functions in momentum space, and the projective representations of the little group at each $\bm{k}$-point determine the degeneracies and topological properties of energy bands. These representations at different $\bm{k}$-points are further constrained by compatibility relations, which limit the possible combinations across the Brillouin zone~\cite{PhysRevLett.117.096404,PhysRevX.7.041069}.
This information can be systematically incorporated into the frameworks of symmetry indicators~\cite{SI} and topological quantum chemistry~\cite{TQC}--- formalisms widely employed in high-throughput searches for topological insulators and semimetals.
Moreover, the projective representation underlies modern, discrete‐symmetry versions of the Lieb-Schultz-Mattis theorem. If the on-site degrees of freedom within a single unit cell furnish a non-trivial projective representation of the global symmetry group --- so that no one-dimensional linear representation exists in the corresponding cohomology class --- then a gapped, symmetry-preserving, and non-degenerate ground state is impossible. The system must either exhibit ground-state degeneracy or support gapless excitations.
This extends the original argument for half-odd-integer spin chains.
~\cite{PhysRevLett.119.127202,OgataTasaki,PhysRevX.8.011040, Cheng2016translational, Cheng2023LSM}.

Despite their importance, detailed and accessible expositions of the mathematics underlying projective representations remain limited.
Classical references on group theory in physics, such as Refs.~\onlinecite{Tinkham,bradley2009mathematical}, typically discuss this topic only briefly, while mathematically oriented texts e.g.~Refs.~\onlinecite{Curtis,Haggarty,Karpilovsky} may present a high barrier for many physics-trained readers.
The purpose of this note is thus to provide an elementary and self-contained introduction to projective representations of \emph{finite} groups, beginning from fundamental definitions.

Throughout this note, we assume $G$ to be a finite group and set $U(1)\coloneqq \{ z\in\mathbb{C}\mid|z|=1 \}$.
As a recurring example, we frequently refer to the simplest nontrivial case,
\begin{align}
G=\mathbb{Z}_2\times\mathbb{Z}_2,
\end{align}
which is isomorphic to the group generated by $\pi$-rotations about the $x$, $y$, and $z$ axes, and thus commonly emerges as an actual symmetry in physics.
For this group, there exist four distinct one-dimensional ordinary (linear) irreducible representations, while there is only a single two-dimensional irreducible projective representation.
In fact, when the group $G$ is \emph{Abelian}, the number of inequivalent irreducible projective representations is always smaller than the number of ordinary (linear) irreducible representations. 

The situation significantly changes for a \emph{non‑Abelian} group.  There it is possible that the number of linear irreducible representations matches the number of the projective ones. This happens when the projective representation corresponds to a $2$‑cocycle $\alpha$ that is non‑trivial in cohomology yet symmetric on every commuting pair, i.e., $\alpha(g,g')=\alpha(g',g)$ for all $g,g'\in G$ with $gg'=g'g$.
Such projective representations correspond to elements of the Bogomolov multiplier $B(G)$, which is a specific subgroup of group cohomology~\cite{Bogomolov,Moravec,Davydov2014}.  

The second aim of this note is to study the role of the Bogomolov multiplier in topological phases of many-body systems. Some of the results presented in this note are known, while some are new.
We first review the properties of symmetry-protected topological (SPT) phases in (1+1)D characterized by a cocycle in the Bogomolov multiplier. Specifically, the edge states of the (1+1)D SPT phase transform under the projective representation characterized by a nontrivial Bogomolov multiplier. In particular, while a large class of (1+1)D SPT phases are characterized by non-local order parameters known as string order parameters, we will see that the SPT phases associated with the Bogomolov multiplier cannot be detected by string order parameters. 
This observation was first made in Ref.~\onlinecite{Pollmann2012detection} using an example of a symmetry group with nontrivial Bogomolov multiplier. 

We then study the implications of Bogomolov multipliers in (1+1)D gapped phases with non-invertible $\text{Rep}(G)$ global symmetry, where $G$ is a non-Abelian finite group with nontrivial Bogomolov multiplier $B(G)$ (see e.g., \cite{McGreevy:2022oyu,schafernameki2023ictp,shao2024tasi,brennan2023introduction,bhardwaj2023lectures} for recent reviews on non-invertible symmetries). In particular, it was recently pointed out in \cite{kobayashi2025soft} that the Bogomolov multiplier gives rise to \textit{distinct} $\text{Rep}(G)$ gapped phases in (1+1)D where $\text{Rep}(G)$ symmetry is spontaneously broken completely. We explicitly provide a pair of lattice models for these $\text{Rep}(G)$ broken phases. These models are  simply obtained by gauging the $G$ symmetry of the $G$ SPT phases associated with a nontrivial Bogomolov multiplier.

We investigate how to distinguish and characterize the two phases in (1+1)D with fully broken $\text{Rep}(G)$ symmetry. It turns out that there are two ways to distinguish them:
\begin{itemize}
    \item While these phases have the same number of local order parameters $O(C)$ labeled by a conjugacy class of a symmetry group $C\in C(G)$, the fusion rules of the local order parameters differ. That is, these two phases are distinguished by phases of fusion coefficients $N^{C''}_{C,C'}$ that appear in the fusion rules
    \begin{align}
    O(C)\times O(C')= \sum_{C''}N^{C''}_{C,C'}O(C'')~.
    \end{align}
    \item Surprisingly, these two distinct SSB phases with fully broken symmetries host nontrivial gapped \textit{interface modes} at the domain wall between two phases. Specifically, consider a system defined on a closed ring, with one SSB phase supported on an interval and the other on its complement. This configuration gives rise to a pair of interfaces. The presence of these interfaces leads to additional ground state degeneracy arising from the interface modes, as well as symmetry actions that differ from those in systems without interfaces. 
    For example, we construct a lattice model realizing a pair of $\text{Rep}(G)$-broken phases and their interfaces, for a specific choice of group $G$ with nontrivial Bogomolov multiplier $B(G)$. In this model, the ground state degeneracy on a ring increases from 32 in the absence of interfaces to 56 when interfaces are present. 
    \end{itemize}

We then review the role of Bogomolov multiplier in (2+1)D topological order described by $G$ gauge theory. In particular, a nontrivial element of Bogomolov multiplier $B(G)$ leads to a global symmetry of (2+1)D $G$ gauge theory with exotic properties 
dubbed a \textit{soft symmetry}, which neither permutes anyons, nor is associated with any symmetry fractionalization~\cite{kobayashi2025soft}. Nevertheless, the soft symmetry is a faithful symmetry that acts on the ground state Hilbert space of topological order on a higher genus surface, and acts on fusion vertices of anyons.  

This (2+1)D perspective gives an additional insight to the Rep$(G)$ broken phases through the framework of symmetry TQFT~\cite{Ji:2019jhk,Kaidi:2022cpf,Freed:2022qnc,Thorngren:2019iar,Lichtman2021,
Gukov:2020btk,Kong2020algebraic,Aasen:2016dop,Chatterjee2023shadow,Moradi2023holography,kaidi2023symmetrytftanomalies,bhardwaj2023charges,
Apruzzi2023symTFT} --- a method that describes (1+1)D gapped phases in terms of intervals of (2+1)D topological order bounded by gapped boundary conditions. 
This provides a convenient framework for classifying phases of matter with a given non-invertible symmetry~\cite{Thorngren:2019iar, Kong2020algebraic, chatterjee2023holographic, Bhardwaj2025clubsandwich, bhardwaj2024hasse,bhardwaj2024gappedphases21dnoninvertible, Antinucci2025gaplessSPT, bottini2025haagerup, bhardwaj2025gappedphases21dnoninvertible, bhardwaj2025gaplessphases21dnoninvertible}.
In particular, a pair of fully symmetry-broken phases correspond to a pair of gapped boundaries of a (2+1)D $G$ gauge theory with identical sets of condensed particles, made possible by soft symmetry in $G$ gauge theory. Using symmetry TQFT, we demonstrate that (1+1)D gapped phases with maximally broken $\text{Rep}(G)$ symmetry are classified by the Bogomolov multiplier $B(G)$.
We further compute the ground state degeneracy of these $\text{Rep}(G)$-broken phases with interfaces within the symmetry TQFT framework, and confirm consistency with the degeneracy of interface modes observed in the lattice models.

This note is organized as follows. In Sec.~\ref{sec:projective}, we review the basics of projective representations.
In Sec.~\ref{sec:Bogomolov}, we discuss several examples of Bogomolov multipliers. In Sec.~\ref{sec:physics} we discuss physical applications of Bogomolov multipliers in (1+1)D and (2+1)D gapped phases. Readers primarily interested in the physical results may begin with Sec.~\ref{sec:physics}.

\section{Projective Representations}
\label{sec:projective}
In this section, we summarize the definitions and basic results on projective representations.
\subsection{Properties of 2-Cocycles}
\subsubsection{Group cohomology}\label{sec:cohomology}
A map $\alpha\colon G\times G\to U(1)$ is called a \textbf{2-cocycle} if it satisfies the \textbf{cocycle condition}
\begin{align}
  \alpha(g,g') \alpha(gg',g'')
  =\alpha(g,g'g'') \alpha(g',g'')
  \label{cocycle}
\end{align}
for all $g,g',g''\in G$.  
The set of all such maps is denoted $Z^{2}(G,U(1))$.  
The trivial 2-cocycle $\alpha^{(0)}$ is defined as $\alpha^{(0)}(g,g')\equiv1$ for all $g,g'\in G$.

A map $\beta\colon G\to U(1)$ is called a \textbf{1-cochain}. Using a 1-cochain $\beta$, we define a map $\delta\beta\colon G\times G\to U(1)$ by
\begin{align}
\delta\beta(g,g')\coloneqq \beta(g)\beta(g')/\beta(gg').
\end{align}
This map $\delta\beta$ is called a \textbf{2-coboundary}; it also satisfies the cocycle condition \eqref{cocycle} and thus belongs to $Z^2(G,U(1))$.
The set of all 2-coboundaries is denoted by $B^2(G,U(1))$.

Two $2$-cocycles $\alpha,\alpha'$ are said to be \textbf{equivalent}, written $\alpha'\sim\alpha$, if there exists a 1-cochain $\beta$ such that
\begin{align}
  \alpha'(g,g')=\delta\beta(g,g') \alpha(g,g').
  \label{equiv}
\end{align}
The equivalence class containing $\alpha$ is denoted by $[\alpha]$.
The set of equivalence classes of 2-cocycles under this relation forms an abelian group called the \textbf{second group cohomology} of $G$ with coefficients in $U(1)$, defined by
\begin{align}
H^2(G,U(1)) \coloneqq Z^2(G,U(1))/B^2(G,U(1)),
\end{align}
which always forms an abelian group.

\begin{example}\label{ex:z2z2}
For $G=\mathbb{Z}_2\times\mathbb{Z}_2$, group elements can be represented as $g=(n,m)$ with $n,m=0,1$, and the group multiplication is defined by $(n,m)(n',m')=(n+n',m+m')\mod 2$.
In this case, we have $H^2(G,U(1))\cong\mathbb{Z}_2$, and a representative nontrivial 2-cocycle not equivalent to $\alpha^{(0)}$ is given by
\begin{align}
\alpha(g,g')\coloneqq(-1)^{nm'}. \label{egomega}
\end{align}
\end{example}

\begin{lemma}
\label{omegasim}
For an element $x\in G$, define $\alpha'(g,g')\coloneqq\alpha(xgx^{-1},xg'x^{-1})$. Then, $\alpha'\sim\alpha$.
\end{lemma}

\begin{proof}
Define a 1-cochain by $\beta(g)\coloneqq \alpha(xgx^{-1},x)/\alpha(x,g)$. Using the cocycle condition \eqref{cocycle}, we obtain
\begin{align}
\delta\beta(g,g')\alpha(g,g')&= \frac{\beta(g)\beta(g')}{\beta(gg')}\alpha(g,g')\notag\\
&= \frac{\alpha(xgx^{-1},x)\alpha(xg'x^{-1},x)[\alpha(x,gg')\alpha(g,g')]} {\alpha(x,g)\alpha(x,g')\alpha(xgg'x^{-1},x)}\notag\\
&= \frac{[\alpha(xgx^{-1},x)\alpha(xg,g')]\alpha(xg'x^{-1},x)} {\alpha(x,g')\alpha(xgg'x^{-1},x)}\notag\\
&= \frac{[\alpha(xgx^{-1},xg')\alpha(xg'x^{-1},x)]} {\alpha(xgg'x^{-1},x)}\notag\\
&=\alpha(xgx^{-1},xg'x^{-1})\notag\\
&=\alpha'(g,g'),
\end{align}
thus verifying \eqref{equiv}.
\end{proof}

\subsubsection{$\alpha$-regularity of conjugacy classes}
We now review some basic facts about conjugacy classes. 
The orbit of an element $g\in G$ under conjugation $g\mapsto xgx^{-1}$ is called the \textbf{conjugacy class} of $g$, denoted by
\begin{align}
C(g)\coloneqq \{xgx^{-1}\mid x \in G\}.
\end{align}
The set of all conjugacy classes of $G$
is denoted by
\begin{align}
    C(G):=\{C(g)\mid g\in G\}. 
\end{align}
The set of elements commuting with $g$ is called the \textbf{centralizer group} of $g$ and denoted by 
\begin{align}
Z(g)\coloneqq \{h \in G\mid hg=gh\}.
\end{align}
The set of elements commuting with all elements of $G$ is called the \textbf{center} of $G$ and denoted by 
\begin{align}
Z(G)\coloneqq \{h \in G\mid hg=gh,\;\forall g\}.
\end{align}
For a subgroup $H$ of $G$, the set of all left cosets $gH$ is denoted by $G/H$. 
\begin{screen}
\begin{theorem}\label{phionetoone}
The map
\begin{align}
  \phi\colon G/Z(g)\longrightarrow C(g),\qquad
  \phi(xZ(g))\coloneqq xgx^{-1},
\end{align}
is a bijection.  
Consequently, the following relation holds among the orders of these groups:
\begin{align}
  |Z(g)| |C(g)|=|G|.
  \label{ZCG}
\end{align}
\end{theorem}
\end{screen}
\begin{proof}
If $xZ(g)=yZ(g)$ then $y^{-1}x\in Z(g)$, so $y^{-1}xg=gy^{-1}x$ and hence $xgx^{-1}=ygy^{-1}$.
Thus the value of $\phi(xZ(g))$ does not depend on the representative of the coset. First assume $\phi(xZ(g))=\phi(yZ(g))$, i.e. $xgx^{-1}=ygy^{-1}$.  
Rewriting gives $y^{-1}xg=gy^{-1}x$, so $y^{-1}x\in Z(g)$ and therefore $xZ(g)=yZ(g)$. Hence $\phi$ is injective. Conversely, for any $h\in C(g)$ there exists $x\in G$ with $h=xgx^{-1}$ by definition of the conjugacy class, and then $\phi(xZ(g))=h$. Hence $\phi$ is surjective.
\end{proof}

We now introduce the notion of $\alpha$-regularity, a key concept in  classifying projective representations. An element $g\in G$ is called \textbf{$\alpha$-regular} if  it satisfies
\begin{align}
\alpha(g,h)=\alpha(h,g)
\end{align}
for all $h\in Z(g)$.
Clearly, $\alpha$-regularity is invariant under conjugation; that is,
\begin{lemma}
If $g$ is $\alpha$-regular, then every conjugate $xgx^{-1}\in C(g)$ is also $\alpha$-regular.
\end{lemma}

\begin{proof}
An element of $Z(xgx^{-1})$ can be written $xhx^{-1}$ with $h \in Z(g)$.  
By Lemma \ref{omegasim},
\begin{align}
  \alpha(xgx^{-1},xhx^{-1})
  =\frac{\beta(g)\beta(h)}{\beta(gh)} \alpha(g,h).
\end{align}
Since $g$ is $\alpha$-regular, the right-hand side is symmetric under $g\leftrightarrow h$, and the claim follows.
\end{proof}

We denote by $C^{(\alpha)}(G)$ the set of $\alpha$-regular conjugacy classes of $G$, and write $|C^{(\alpha)}(G)|$ for the number of such classes.

\begin{lemma}\label{Zgrouphomo}
For $g\in G$, define a map $\varphi_g\colon Z(g)\to U(1)$ by
\begin{align}
\varphi_g(h)\coloneqq \frac{\alpha(g,h)}{\alpha(h,g)}.
\label{eq:slant}
\end{align}
Then $\varphi_g$ is a group homomorphism.
\end{lemma}

\begin{proof}
Cocycle condition \eqref{cocycle} immediately yields $\varphi_g(hh')=\varphi_g(h)\varphi_g(h')$ for $h,h' \in Z(g)$. 
\end{proof}

\begin{screen}
\begin{theorem}[Number of $\alpha$-regular conjugacy classes]\label{thm:regclasscount}
The number of $\alpha$-regular conjugacy classes $|C^{(\alpha)}(G)|$ is given by
\begin{align}
  |C^{(\alpha)}(G)|
  =\frac{1}{|G|}\sum_{g,g'\in G}
    \frac{\alpha(g,g')}{\alpha(g',g)} 
    \delta_{gg', g'g}.
  \label{regularClGcount}
\end{align}
Here, $\delta_{x,y}=1$ if $x=y$ and $0$ otherwise.  
For the trivial cocycle $\alpha^{(0)}$, this reduces to the total number of conjugacy classes
\begin{align}
  |C(G)|=\frac{1}{|G|}
  \sum_{g,g'\in G}\delta_{gg', g'g}.
  \label{ClGcount}
\end{align}
\end{theorem}
\end{screen}

\begin{proof}
By Lemma \ref{Zgrouphomo} the map $\varphi_g$ gives a one-dimensional representation of $Z(g)$.  
If $g$ is $\alpha$-regular, $\varphi_g$ is trivial and
\begin{align}
  \sum_{h\in Z(g)}\varphi_g(h)=|Z(g)|.
\end{align}
If $g$ is \emph{not} $\alpha$-regular, character orthogonality implies the sum vanishes.  
Hence
\begin{align}
  \sum_{h\in Z(g)}
  \frac{\alpha(g,h)}{\alpha(h,g)}
  =
  \begin{cases}
    |Z(g)| & (\text{If $g$ is $\alpha$-regular})\\
    0      & (\text{otherwise}).
  \end{cases}
  \label{regularClG}
\end{align}
Using $|Z(g)|=|Z(g_0)|=|G|/|C(g_0)|$ for $g\in C(g_0)$ from \eqref{ZCG} and summing over each conjugacy class yields Eq.~\eqref{regularClGcount}.
\end{proof}

It follows immediately from Eq.~\eqref{regularClGcount} that $|C^{(\alpha)}(G)|=|C(G)|$ iff \emph{all} elements of $G$ are $\alpha$-regular.

\begin{example}
For the Abelian group $G=\mathbb{Z}_2\times\mathbb{Z}_2$, one has $Z(g)=G$ and $C(g)=\{g\}$.  
With the cocycle \eqref{egomega},
\begin{align}
|C^{(\alpha)}(G)|=\frac14\sum_{n,n',m,m'=0}^{1}(-1)^{nm'-n'm}=1,
\end{align}
so the only $\alpha$-regular conjugacy class is $C(e)=\{e\}$.
\end{example}


\subsection{Projective Representations}
\subsubsection{Multipliers and Projective Representations}
From the cocycle condition, 
we see that for all $g\in G$, a general 2-cocycle $\tilde{\alpha}\colon G\times G\to U(1)$ satisfies $\tilde{\alpha}(g,e)=\tilde{\alpha}(e,g)=\tilde{\alpha}(e,e)$.  Thus, by defining a constant 1-cochain $\beta(g)\equiv 1/\tilde{\alpha}(e,e)$, the adjusted cocycle \begin{align} \alpha(g,g')\coloneqq\delta\beta(g,g')\tilde{\alpha}(g,g')=\frac{\tilde{\alpha}(g,g')}{\tilde{\alpha}(e,e)} \end{align} satisfies the normalization condition \begin{align} \alpha(e,g)=\alpha(g,e)=1 \end{align} for all $g\in G$. A 2-cocycle $\alpha\colon G\times G\to U(1)$ with this normalization condition is called a \textbf{multiplier system} of $G$.


Given a multiplier system $\alpha$, a map $U\colon G\to U(d)$ satisfying
\begin{align}
U(g)U(g')=\alpha(g,g')U(gg')
\end{align}
for all $g,g'\in G$ is called a \textbf{projective representation} of $G$ associated with $\alpha$. Ordinary (linear) representations correspond to the trivial multiplier system $\alpha=\alpha^{(0)}$. The dimension $d$ of the matrices is called the \textbf{dimension of the representation}, and $\chi_U(g)\coloneqq\mathrm{tr}U(g)$ is called the \textbf{character}.

\begin{example}
For $G=\mathbb{Z}_2\times\mathbb{Z}_2$, a projective representation associated with the multiplier $\alpha(g,g')=(-1)^{nm'}$ is given by 
\begin{align}
U(0,0)=\sigma_0, \quad U(1,0)=\sigma_1, \quad U(0,1)=\sigma_2, \quad U(1,1)=-i\sigma_3, \label{Z2Z2rep1}
\end{align}
where $\sigma_\mu$ are the Pauli matrices. This representation has dimension $2$, and its character is $2$ for the element $(0,0)$ and $0$ for all other elements.
\end{example}

\begin{theorem}
\label{theorem3}
Let $U$ be a projective representation of $G$ associated with $\alpha$. If an element $g\in G$ is not $\alpha$-regular, then
\begin{align}
\chi_U(g)=0.
\end{align}
\end{theorem}

\begin{proof}
For $h\in Z(g)$,
\begin{align}
U(g)U(h)=\alpha(g,h)U(gh),\quad U(h)U(g)=\alpha(h,g)U(hg).
\end{align}
Since $gh=hg$, we have
\begin{align}
\alpha(g,h)U(h)U(g)U(h)^{-1}=\alpha(h,g)U(g).
\end{align}
Taking the trace, we obtain the result when $\alpha(g,h)\neq\alpha(h,g)$.
\end{proof}

When an element $g \in G$ is $\alpha$-regular, it is possible to have $\chi_U(g) \neq 0$. However, note that in projective representations, even if elements $g$ and $g'$ belong to the same conjugacy class, it \emph{does not necessarily} hold that $\chi_U(g')=\chi_U(g)$. Indeed, if $g'=xgx^{-1}$, then
\begin{align}
\chi_U(g)=\frac{\alpha(x,g)\alpha(xg,x^{-1})}{\alpha(x,x^{-1})}\chi_U(g'),  \label{ggpchi}
\end{align}
and thus a discrepancy by a $U(1)$ factor arises.
\begin{screen}
\begin{theorem}
If there exists a projective representation of dimension $d$, then $\alpha^d\sim\alpha^{(0)}$. In particular, if there is a one-dimensional projective representation, then $\alpha\sim\alpha^{(0)}$.
\end{theorem}
\end{screen}
\begin{proof}
Taking the determinant of both sides of the defining relation for a projective representation,
\begin{align}
U(g)U(g')=\alpha(g,g')U(gg'),
\end{align}
we have
\begin{align}
\alpha(g,g')^d=\frac{\det U(g) \det U(g')}{\det U(gg')}.
\end{align}
Thus, by defining a 1-cochain $\beta(g)=\det U(g)\in U(1)$, we see $\alpha(g,g')^d=\delta\beta(g,g')$, proving the claim.
\end{proof}

As we will see later, a regular projective representation, whose dimension is $|G|$, always exists for any multiplier system $\alpha$. Thus,
\begin{align}
\alpha^{|G|}\sim\alpha^{(0)}.\label{omegad}
\end{align}

\subsubsection{Schur's Lemma}
A projective representation $U$ (with a fixed multiplier $\alpha$) is
\textbf{reducible} if there exists non-trivial projective
representations $V$, $W$ of equal multiplier such that, in a common
basis,
\begin{align}
  U(g)=
  \begin{pmatrix}
    V(g) & 0\\
    0    & W(g)
  \end{pmatrix}
  \qquad(\forall g\in G).
\end{align}
Otherwise $U$ is \textbf{irreducible}.  
By definition every one-dimensional representation is irreducible.
Irreducible representations will be denoted $U_A$ with dimension
$d_A$.

Let $U_A$ and $U_B$ be irreducible projective representations associated with the same multiplier system $\alpha$. A matrix $M$ satisfying
\begin{align}
U_A(g)M=MU_B(g)
\end{align}
for all $g\in G$ is either a zero matrix or an invertible square matrix with dimension $d_A=d_B$. Moreover, if $A=B$, then $M$ must be proportional to the identity matrix.

\begin{screen}
\begin{theorem}[Orthogonality Relations for Representations]
\label{thmorthogonalrep}
Let $U_A$ and $U_B$ be irreducible projective representations of $G$ associated with $\alpha$. Then the following relation holds:
\begin{align}
\sum_{g\in G}[U_A(g)]_{ij}[U_B(g)]_{kl}^*=\frac{|G|}{d_A}\delta_{AB}\delta_{ik}\delta_{jl}.\label{orthogonalrep}
\end{align}
\end{theorem}
\end{screen}

\begin{proof}
Consider an arbitrary $d_A\times d_B$ matrix $X$ and define
\begin{align}
M\coloneqq\sum_{g\in G}U_A(g)XU_B(g)^\dagger.
\end{align}
Then, for any $g'\in G$,
\begin{align}
U_A(g')MU_B(g')^\dagger&=\sum_{g\in G}U_A(g')U_A(g)XU_B(g)^\dagger U_B(g')^\dagger=\sum_{g\in G}[U_A(g')U_A(g)]X[U_B(g')U_B(g)]^\dagger\notag\\
&=\sum_{g\in G}[\alpha(g',g)U_A(g'g)]X[\alpha(g',g)U_B(g'g)]^\dagger=\sum_{g''\in G}U_A(g'')XU_B(g'')^\dagger\notag\\
&=M,
\end{align}
thus,
\begin{align}
U_A(g')M=MU_B(g').\label{UAMMUB}
\end{align}
By Schur's lemma, $M=0$ if $A\ne B$, and $M$ is proportional to the identity matrix if $A=B$. Choose $X$ such that its $(j,l)$ element is 1 and all other elements are 0, i.e., $[X]_{j'l'}=\delta_{j'j}\delta_{l'l}$. If $B\neq A$, we have
\begin{align}
\sum_{g\in G}[U_A(g)]_{ij'}[X]_{j'l'}[U_B(g)^\dagger]_{l'k}=\sum_{g\in G}[U_A(g)]_{ij}[U_B(g)^*]_{kl}=0.
\end{align}
For $B=A$, we have
\begin{align}
\sum_{g\in G}[U_A(g)]_{ij'}[X]_{j'l'}[U_A(g)^\dagger]_{l'k}=\sum_{g\in G}[U_A(g)]_{ij}[U_A(g)^*]_{kl}=c_{jl}\delta_{ik}.
\end{align}
Summing over $i$ with $k=i$, we obtain $c_{jl}=(|G|/d_A)\delta_{jl}$.
\end{proof}

\begin{screen}
\begin{theorem}[Orthogonality Relations for Characters 1]
Let $U_A$ ($A=1,2,\cdots,N_{\mathrm{rep}}^{(\alpha)}$) be irreducible projective representations of $G$ associated with $\alpha$, and let their characters be denoted by $\chi_{U_A}$. Suppose $C_I$ ($I=1,2,\cdots,|C^{(\alpha)}(G)|$) are the $\alpha$-regular conjugacy classes with representatives $g_I\in C_I$. Then, the following relation holds:
\begin{align}
\sum_{I=1}^{|C^{(\alpha)}(G)|}|C_I|\chi_{U_A}(g_I)^*\chi_{U_B}(g_I)=|G|\delta_{AB}.\label{orthogonalcharacter1}
\end{align}
\end{theorem}
\end{screen}

\begin{proof}
Setting $j=i$ and $l=k$ in the orthogonality relation for representations \eqref{orthogonalrep} and summing over $i,k$, we obtain
\begin{align}
\sum_{g\in G}\chi_{U_A}(g)^*\chi_{U_B}(g)=|G|\delta_{AB}.
\end{align}
From Theorem \ref{theorem3}, only $\alpha$-regular elements $g$ contribute to this sum. Furthermore, although for $g\in C_I$, we generally do not have $\chi_{U_A}(g)=\chi_{U_A}(g_I)$ exactly, by Eq.~\eqref{ggpchi} the difference is at most a phase factor, thus leading to the stated relation.
\end{proof}

\subsection{Group Algebra}
\subsubsection{Group Algebra and Projective Representations}
To obtain further information about irreducible representations, we introduce the group algebra. Fixing a multiplier system $\alpha$, consider an $|G|$-dimensional vector space over complex numbers. Suppose the linearly independent basis vectors of this vector space are given by $\{|g\rangle\}_{g\in G}$ indexed by elements $g\in G$. Define the multiplication on this vector space by
\begin{align}
|g\rangle\cdot|g'\rangle\coloneqq\alpha(g,g')|gg'\rangle. \label{eq:twistedmul}
\end{align}
The identity element is $|e\rangle$, satisfying $|e\rangle\cdot|g\rangle=|g\rangle\cdot|e\rangle=|g\rangle$. This algebra is called the \textbf{$\alpha$-twisted group algebra}.

The operation of left multiplication by $|g\rangle$, denoted by $L_g$, is a linear map on this vector space:
\begin{align}
L_g\Big(\sum_{g''\in G}|g''\rangle c_{g''}\Big)&\coloneqq\sum_{g''\in G}|g\rangle\cdot|g''\rangle c_{g''}\notag\\
&=\sum_{g''\in G}\alpha(g,g'')|gg''\rangle c_{g''}\notag\\
&=\sum_{g'\in G}|g'\rangle\sum_{g''\in G}[U_{\mathrm{reg}}(g)]_{g',g''}c_{g''}
\end{align}
Here, we define a $|G|$-dimensional matrix $U_{\mathrm{reg}}(g)$ with components given by
\begin{align}
[U_{\mathrm{reg}}(g)]_{g',g''}=\delta_{g',gg''}\alpha(g,g'')
\end{align}
The matrix $U_{\mathrm{reg}}(g)$ satisfies
\begin{align}
[U_{\mathrm{reg}}(g)^\dagger U_{\mathrm{reg}}(g)]_{g',g''}=\delta_{g',g''}
\end{align}
thus it is unitary. Moreover, it fulfills
\begin{align}
U_{\mathrm{reg}}(g)U_{\mathrm{reg}}(g')=\alpha(g,g')U_{\mathrm{reg}}(gg'),
\end{align}
making it a projective representation associated with the multiplier system $\alpha$. This representation is called the \textbf{regular projective representation}. Its character is given by
\begin{align}
\chi_{U_{\mathrm{reg}}}(g)=|G|\delta_{g,e}. \label{tracereg1}
\end{align}

\begin{screen}
\begin{theorem}[Sum of Squares of Dimensions]
Let $U_A$ ($A=1,2,\dots,N_{\mathrm{rep}}^{(\alpha)}$) be irreducible projective representations associated with $\alpha$, and let their dimensions be $d_A$. Then, the following relation holds:
\begin{align}
\sum_{A=1}^{N_{\mathrm{rep}}^{(\alpha)}}d_A^2=|G|.\label{repdim}
\end{align}
\end{theorem}
\end{screen}

\begin{proof}
Interpreting the orthogonality relation for representations \eqref{orthogonalrep} as the orthogonality of $\sum_{A=1}^{N_{\mathrm{rep}}^{(\alpha)}}d_A^2$ vectors of dimension $|G|$, we have
\begin{align}
\sum_{A=1}^{N_{\mathrm{rep}}^{(\alpha)}}d_A^2\leq |G|.\label{repdim2}
\end{align}
On the other hand, using the orthogonality of characters, if the regular projective representation contains an irreducible representation $U_A$, it does so exactly $d_A$ times. Therefore, we have $U_{\mathrm{reg}}(g)=\bigoplus_{A}d_A U_A(g)$, and taking the trace yields
\begin{align}
\chi_{U_{\mathrm{reg}}}(g)=\sum_{A}d_A \chi_{U_A}(g).\label{tracereg2}
\end{align}
From \eqref{tracereg1}, setting $g=e$ gives
\begin{align}
\sum_{A}d_A^2=|G|,
\end{align}
which achieves the upper bound of \eqref{repdim2}. Thus, there are no other irreducible representations, and each irreducible representation is included exactly $d_A$ times in the regular projective representation.
\end{proof}

\begin{example}
The group $G=\mathbb{Z}_2\times\mathbb{Z}_2$ has four one-dimensional linear representations
\begin{align}
U(0,0)=1, \quad U(1,0)=\xi_1, \quad U(0,1)=\xi_2, \quad U(1,1)=\xi_1\xi_2\label{Z2Z2rep2}
\end{align}
with $\xi_1,\xi_2=\pm1$. Meanwhile, if we assume the multiplier system $\alpha$ given by \eqref{egomega}, there is only one irreducible projective representation of dimension two in Eq.~\eqref{Z2Z2rep1}. In either case, the relation \eqref{repdim} holds, as $1^2\times 4=4$ and $2^2\times 1=4$.
\end{example}

\subsubsection{Products of Conjugacy Classes and Structure Constants}
\begin{lemma}
\label{structure}
For $g\in G$, define
\begin{align}
|\Psi_g\rangle\coloneqq\sum_{x\in G}\frac{1}{\alpha(x,x^{-1})}|x\rangle\cdot|g\rangle\cdot|x^{-1}\rangle
=\sum_{x\in G}\frac{\alpha(x,g)\alpha(xg,x^{-1})}{\alpha(x,x^{-1})}|xgx^{-1}\rangle.
\end{align}
Then, for all $g'\in G$,
\begin{align}
|g'\rangle\cdot|\Psi_g\rangle=|\Psi_g\rangle\cdot|g'\rangle.\label{ringcenter}
\end{align}
\end{lemma}

\begin{proof}
We have
\begin{align}
|g'\rangle\cdot|\Psi_g\rangle&=\sum_{x\in G}\frac{\alpha(g',x)}{\alpha(x,x^{-1})}|g'x\rangle\cdot|g\rangle\cdot|x^{-1}\rangle\notag\\
&=\sum_{y\in G}\frac{\alpha(g',g'{}^{-1}y)}{\alpha(g'{}^{-1}y,y^{-1}g')}|y\rangle\cdot|g\rangle\cdot|y^{-1}g'\rangle\notag\\
&=\sum_{y\in G}\frac{\alpha(g',g'{}^{-1}y)}{\alpha(g'{}^{-1}y,y^{-1}g')\alpha(y^{-1},g')}|y\rangle\cdot|g\rangle\cdot|y^{-1}\rangle\cdot|g'\rangle\notag\\
&=|\Psi_g\rangle\cdot|g'\rangle.\label{derivation}
\end{align}
The second equality follows by the substitution $y=g'x$, and the last equality uses the cocycle condition on $(g',g'{}^{-1}y,y^{-1})$, $(g'{}^{-1}y,y^{-1},g')$, and $(g',g'{}^{-1},g')$:
\begin{align}
\frac{\alpha(g',g'{}^{-1}y)}{\alpha(g'{}^{-1}y,y^{-1}g')\alpha(y^{-1},g')}=\frac{1}{\alpha(y,y^{-1})}.
\end{align}
\end{proof}
An element $|\Psi_g\rangle$ satisfying \eqref{ringcenter} is called the center of the group algebra. Writing $Q=G/Z(g)$, each element of $G$ can be expressed as $qh$ ($q\in Q$, $h\in Z(g)$). Thus,
\begin{align}
|\Psi_g\rangle=\sum_{h\in Z(g)}\sum_{q\in Q}\frac{\alpha(qh,g)\alpha(qhg,h^{-1}q^{-1})}{\alpha(qh,h^{-1}q^{-1})}|qgq^{-1}\rangle=\sum_{h\in Z(g)}\frac{\alpha(h,g)}{\alpha(g,h)}\sum_{q\in Q}\frac{\alpha(q,g)\alpha(qg,q^{-1})}{\alpha(q,q^{-1})}|qgq^{-1}\rangle.\label{Sgtrans}
\end{align}
Applying \eqref{regularClG} to the sum over $h$, we find $|\Psi_g\rangle=0$ unless $g$ is $\alpha$-regular.

\begin{lemma}
\label{structure2}
Let $C_I$ ($I=1,2,\dots,|C^{(\alpha)}(G)|$) be the $\alpha$-regular conjugacy classes of $G$, with representatives $g_I\in C_I$. Then,
\begin{align}
|\Psi_{g_I}\rangle\cdot|\Psi_{g_J^{-1}}\rangle=\sum_K c_{IJ}^{K}|\Psi_{g_K}\rangle.\label{eq:structure}
\end{align}
Moreover, if the conjugacy class consisting solely of the identity element $e$ is denoted by $C_1$, we have
\begin{align}
c_{IJ}^{K=1}=\alpha(g_I,g_I^{-1})|Z(g_I)|\delta_{IJ}.\label{cIJK}
\end{align}
\end{lemma}

\begin{proof}
Equation \eqref{eq:structure} follows directly since the elements $|\Psi_{g_I}\rangle$ ($I=1,\dots,|C^{(\alpha)}(G)|$) form a basis for the center of the group algebra. Next, we prove \eqref{cIJK}. Firstly, note that $|\Psi_{g_I}\rangle\cdot|\Psi_{g_J^{-1}}\rangle$ contains the element $|\Psi_e\rangle=|G||e\rangle$ only when $I=J$. Moreover,
\begin{align}
(xg_Ix^{-1})(yg_I^{-1}y^{-1})=e\quad\Leftrightarrow\quad x^{-1}y\in Z(g_I),
\end{align}
so in the double summation $\sum_{x,y\in G}$ appearing in $|\Psi_{g_I}\rangle\cdot|\Psi_{g_J^{-1}}\rangle$, we can set $y=xh$ with $h\in Z(g_I)$. Performing transformations analogous to those in \eqref{Sgtrans}, we finally obtain
\begin{align}
\frac{\alpha(x,g_I)\alpha(xg_I,x^{-1})}{\alpha(x,x^{-1})}\frac{\alpha(x,g_I^{-1})\alpha(xg_I^{-1},x^{-1})}{\alpha(x,x^{-1})}=\frac{\alpha(g_I,g_I^{-1})}{\alpha(xg_Ix^{-1},xg_I^{-1}x^{-1})},
\end{align}
from which \eqref{cIJK} follows.
\end{proof}

\begin{screen}
\begin{theorem}[Orthogonality Relations for Characters 2]
Let $U_A$ ($A=1,2,\dots,N_{\mathrm{rep}}^{(\alpha)}$) be irreducible projective representations associated with $\alpha$, and let their characters be denoted by $\chi_{U_A}$. Suppose $C_I$ ($I=1,2,\dots,|C^{(\alpha)}(G)|$) are the $\alpha$-regular conjugacy classes, with representatives $g_I\in C_I$. Then, the following relation holds:
\begin{align}
&\sum_{A=1}^{N_{\mathrm{rep}}^{(\alpha)}}\chi_{U_A}(g_I)\chi_{U_A}(g_J)^*=|Z(g_I)|\delta_{IJ}.\label{orthogonalcharacter2}
\end{align}
\end{theorem}
\end{screen}

\begin{proof}
Define
\begin{align}
M(g)\coloneqq\sum_{x\in G}\frac{1}{\alpha(x,x^{-1})}U_A(x)U_A(g)U_A(x^{-1})=\sum_{x\in G}U_A(x)U_A(g)U_A(x)^\dagger.
\end{align}
From Lemma \ref{structure}, for any $g'\in G$,
\begin{align}
U_A(g')M(g)=M(g)U_A(g').
\end{align}
By Schur's lemma, $M(g)=c(g)U_A(e)$, and calculating the trace yields $c(g)=(|G|/d_A)\chi_{U_A}(g)$. Thus,
\begin{align}
M(g)=\frac{|G|}{d_A}\chi_{U_A}(g)U_A(e).\label{MI}
\end{align}
Lemma \ref{structure2} further gives
\begin{align}
M(g_I) M(g_J^{-1})=\sum_{K=1}^{|C^{(\alpha)}(G)|}c_{IJ}^{K}M(g_K).\label{MMM}
\end{align}
Moreover, from \eqref{tracereg1} and \eqref{tracereg2} for the regular projective representation, we have
\begin{align}
\sum_{A=1}^{N_{\mathrm{rep}}^{(\alpha)}}d_A \chi_{U_A}(g_I)=|G|\delta_{I,1}.\label{tracereg3}
\end{align}
Substituting \eqref{MI} into \eqref{MMM} and multiplying by $d_A$, we obtain
\begin{align}
|G|\chi_{U_A}(g_I)\chi_{U_A}(g_J^{-1})=\sum_{K=1}^{|C^{(\alpha)}(G)|}c_{IJ}^{K}d_A\chi_{U_A}(g_K).
\end{align}
Summing over $A$ and using \eqref{tracereg3}, along with substituting \eqref{cIJK}, we derive \eqref{orthogonalcharacter2} using $\chi_{U_A}(g_I^{-1})=\alpha(g_I,g_I^{-1})\chi_{U_A}(g_I)^*$.
\end{proof}

\begin{screen}
\begin{theorem}[Number of Irreducible Projective Representations]
The number of inequivalent irreducible projective representations associated with $\alpha$ is equal to the number of $\alpha$-regular conjugacy classes. That is,
\begin{align}
N_{\mathrm{rep}}^{(\alpha)}=|C^{(\alpha)}(G)|.\label{NrNC}
\end{align}
\end{theorem}
\end{screen}

\begin{proof}
Equation \eqref{orthogonalcharacter1} implies the orthogonality of $N_{\mathrm{rep}}^{(\alpha)}$ vectors of dimension $|C^{(\alpha)}(G)|$ labeled by the indices $A=1,\dots,N_{\mathrm{rep}}^{(\alpha)}$:
\begin{align}
(\sqrt{|C_1|}\chi_{U_A}(g_1),\sqrt{|C_2|}\chi_{U_A}(g_2),\dots,\sqrt{|C_{|C^{(\alpha)}(G)|}|}\chi_{U_A}(g_{|C^{(\alpha)}(G)|})),
\end{align}
thus giving $|C^{(\alpha)}(G)|\geq N_{\mathrm{rep}}^{(\alpha)}$. Conversely, \eqref{orthogonalcharacter2} indicates the orthogonality of $|C^{(\alpha)}(G)|$ vectors of dimension $N_{\mathrm{rep}}^{(\alpha)}$ labeled by the indices $I=1,\dots,|C^{(\alpha)}(G)|$:
\begin{align}
(\chi_{U_1}(g_I),\chi_{U_2}(g_I),\dots,\chi_{U_{N_{\mathrm{rep}}^{(\alpha)}}}(g_I)),
\end{align}
yielding $|C^{(\alpha)}(G)|\leq N_{\mathrm{rep}}^{(\alpha)}$. Combining these inequalities, we obtain \eqref{NrNC}.
\end{proof}

\subsubsection{Method to Determine Irreducible Representations}
\label{irreduciblemethod}
When the order of group $G$ is large, it becomes difficult to enumerate all irreducible projective representations associated with $\alpha$. Here, we introduce a numerical trick to simplify this task \cite{SI,Yang_2021,2304.01827,PhysRevB.111.134407}.

First, generate a random $|G|$-dimensional Hermitian matrix $H_0$. Define a $|G|$-dimensional Hermitian matrix $H$ using the regular projective representation:
\begin{align}
H\coloneqq \frac{1}{|G|}\sum_{g\in G}U_{\mathrm{reg}}(g)H_0U_{\mathrm{reg}}(g)^\dagger.
\end{align}
Then, similarly to \eqref{UAMMUB}, for any $g'\in G$,
\begin{align}
U_{\mathrm{reg}}(g')H=HU_{\mathrm{reg}}(g').
\end{align}
Thus, $H$ is a Hamiltonian with symmetry $G$. Unless accidental degeneracy occurs, its eigenvalues split according to irreducible projective representations, with each representation $U_A$ producing a degeneracy of dimension $d_A$ repeated exactly $d_A$ times. If $\bm{\phi}_{n\kappa}$ ($\kappa$ distinguishes degeneracies) are eigenvectors corresponding to the $n$-th eigenvalue, defining
\begin{align}
[U_n(g)]_{\kappa\kappa'}\coloneqq\bm{\phi}_{n\kappa}^\dagger U(g)\bm{\phi}_{n\kappa'}
\end{align}
yields irreducible representations. By repeating this for all eigenvalues, all irreducible projective representations can be obtained.

\subsubsection{Central Extension}

Let $A$ be any Abelian group, and consider a 2-cocycle $\omega\in Z^2(G,A)$ taking values in $A$. 
Consider the direct product set $A \times G$ with elements denoted as $(a,g)$, where $a\in A$ and $g\in G$. Define the multiplication on this set by
\begin{align}
(a,g)(a',g')=(aa'\omega(g,g'),gg').
\end{align}
This defines a group $\hat{G}$ called the \textbf{central extension} of $G$. In this construction,
\begin{align}
1\rightarrow A
\xrightarrow{i:a\mapsto(a,e)}
\hat{G}
\xrightarrow{\pi:(a,g)\mapsto g}
G\rightarrow 1
\end{align}
forms a short exact sequence, with $i(A)\subset Z(\hat{G})$.

In fact, the projective representation of $G$ is understood as \emph{a linear representation} of a certain central extension $\hat{G}$. Let us consider a projective representation characterized by a 2-cocycle $\alpha \in Z^2(G,U(1))$. As shown in \eqref{omegad}, by appropriately adjusting the coboundary, we can assume without loss of generality that it takes values in $\mathbb{Z}_{|G|}\subset U(1)$. Then, take $A=\mathbb{Z}_{|G|}$ and $\omega = \alpha\in Z^2(G,A)$ to define the finite group $\hat{G}$.

Now consider the linear representation $\hat{U}$ of $\hat{G}$ satisfying $\hat{U}(a,e)=a\,\mathrm{id}$. Define $U(g)\coloneqq\hat{U}(1,g)$. Then,
\begin{align}
U(g)U(g')=\hat{U}(1,g)\hat{U}(1,g')&=\hat{U}(\alpha(g,g'),gg')=\alpha(g,g')\hat{U}(1,gg')=\alpha(g,g')U(gg').
\end{align}
Thus, $U$ provides a projective representation of $G$. Consequently, some results established previously, such as the orthogonality relations for characters and representations, follow immediately from known results for linear representations.
Moreover, the condition that the projective representation is trivial $\alpha\sim\alpha^{(0)}$ is equivalent to the existence of a group isomorphism $\phi\colon \hat{G}\to A\times G$ defined by $\phi(a,g)=(a\beta(g),g)$ with $\alpha = \delta\beta$.

\section{Bogomolov Multiplier}
\label{sec:Bogomolov}
We now discuss examples where every element $g \in G$ is $\alpha$-regular, but $\alpha$ is not a 2-coboundary. 

Let $G$ be a finite group and $\alpha$ a multiplier of $G$.  The following two conditions are equivalent: (i) For every abelian subgroup $A\subset G$, the restriction $\alpha|_{A\times A}\colon A\times A\to U(1)$ is a coboundary, (ii) For every pair of commuting elements $g,g'\in G$ (so $gg'=g'g$), one has $\alpha(g,g')=\alpha(g',g)$. The cohomology class containing such a cocycle $\alpha$ is called the Bogomolov multiplier \cite{Bogomolov,Moravec, Davydov2014} and is denoted $B(G)\subset H^{2}(G,U(1))$. By definition, $B(G)=0$ for any Abelian groups.

Historically, Mangold asserted that if every $g \in G$ is $\alpha$-regular, then $\alpha$ is a 2-coboundary (meaning $B(G)$ is trivial for any finite group) \cite{Higgs1989projective}. However, as demonstrated in Ref.~\onlinecite{Higgs1989projective}, a counterexample exists. 
Moravec \cite{Moravec} exhaustively investigated groups with $|G| \leq 64$, finding that the smallest order for which $B(G)$ is nontrivial is $|G| = 64$, and among 267 groups of order $|G| = 64$ (excluding isomorphic cases), nine groups have nontrivial $B(G)$.

If $[\alpha]\in B(G)$, the following relations must all hold:
\begin{itemize}
\item Every conjugacy class of $G$ is an $\alpha$-regular conjugacy class. Denote the number of these classes by $|C^{(\alpha)}(G)|=|C(G)|$.
\item Let $d_A$ be the dimension of irreducible projective representations $U_A$ of $G$ associated with $\alpha$. The number of distinct irreducible projective representations, $N_{\mathrm{rep}}^{(\alpha)}$, matches $|C(G)|$, and the relation \eqref{repdim} holds.
\item Let $d_A^{(0)}$ be the dimension of irreducible linear representations $U_A^{(0)}$ of $G$. The number of distinct irreducible linear representations, $N_{\mathrm{rep}}^{(0)}$, also matches $|C(G)|$, and the relation \eqref{repdim} similarly holds.
\end{itemize}

\subsection{Pollmann-Turner Example}
\label{sec:pollmann-turner}

We begin with the simplest constructed example. Pollmann and Turner discussed the following example to construct an SPT phase that cannot be detected by string-order~\cite{Pollmann2012detection}.

Consider the group defined by four generators $a,b,c_{1},c_{2}$ with the relations:
\begin{align}
\begin{split}
& a^{2}=b^{2}=c_{1}^{2}=c_{2}^{2}=1,\\
& [a,b]=[c_{1},c_{2}],\\
& [b,c_{1}]=[b,c_{2}]=1,\\
& [[G,G],G]=1,
\label{eq:pollmann_def}
\end{split}
\end{align}
where the commutator $[p,q]$ is defined by $[p,q]\coloneqq p q p^{-1} q^{-1}$. 
Any element $g \in G$ can be uniquely expressed as
\begin{align}
g=a^{n_a}b^{n_b}c_1^{\ell_1}c_2^{\ell_2}[a,b]^k[a,c_1]^{d_1}[a,c_2]^{d_2},
\end{align}
where $n_a, n_b, \ell_1, \ell_2, k, d_1, d_2$ take values in ${0,1}$. Thus, the group order is $|G|=128$. According to Ref.~\onlinecite{Moravec}, among 2328 groups of order $|G|=128$ (excluding isomorphic ones), 230 have nontrivial $B(G)$.
For simplicity, let us introduce notations for group elements 
\begin{align}
    x := [a,b]=[c_1,c_2],\quad y_1:= [a,c_1],\quad y_2:=[a,c_2]~.
\end{align}
This group is then regarded as a central extension
\begin{align}
    (\Z_2)^3\to G\to (\Z_2)^4~,
\end{align}
where $(\Z_2)^4$ is generated by $a,b,c_1,c_2$, while $(\Z_2)^3$ is generated by $x,y_1,y_2$.
The product of elements
\begin{align}
g=a^{n_a}b^{n_b}c_1^{\ell_1}c_2^{\ell_2}x^k y_1^{d_1}y_2^{d_2}, \quad
g'=a^{n_a'}b^{n_b'}c_1^{\ell_1'}c_2^{\ell_2'}x^{k'}y_1^{d_1'}y_2^{d_2'}
\label{eq:labelgroupelement_PT}
\end{align}
is calculated as
\begin{align}
gg'&=a^{n_a+n_a'}b^{n_b+n_b'}c_1^{\ell_1+\ell_1'}c_2^{\ell_2+\ell_2'}x^{k+k'+n'_an_b+\ell_1'\ell_2}y_1^{d_1+d_1'+n'_a\ell_1}y_2^{d_2+d_2'+n'_a\ell_2}.
\end{align}

\subsubsection{Multiplier system}
Define a multiplier system $\alpha\in Z^2(G,U(1))$ by
\begin{align}
\alpha(g,g')\coloneqq(-1)^{n_an_b'}.
\label{eq:BG Pollmann}
\end{align}
This is normalized and satisfies the cocycle condition:
\begin{align}
\alpha(g,g')\alpha(gg',g'')
=(-1)^{n_an_b'}(-1)^{(n_a+n_a')n_b''}
=(-1)^{n_a(n_b'+n_b'')}(-1)^{n'_an''_b}=\alpha(g,g'g'')\alpha(g',g'').
\end{align}

The commutativity condition $gg'=g'g$ is given by:
\begin{align}
&n_a'n_b+\ell_1'\ell_2=n_an'_b+\ell_1\ell_2',\notag\\
&n_a'\ell_1=n_a\ell_1', n_a'\ell_2=n_a\ell_2'\mod 2.
\end{align}
Under this condition, we find $(-1)^{n_a'n_b}=(-1)^{n_an_b'}$ so that all elements are $\alpha$-regular.

\subsubsection{Projective representations}
Using the method introduced in Section \ref{irreduciblemethod}, we find 32 irreducible projective representations of dimension 2:
\begin{align}
|C^{(\alpha)}(G)|=N_{\mathrm{rep}}^{(\alpha)}=32.
\end{align}
One can verify the relation \eqref{repdim} by
\begin{align}
\sum_{A=1}^{N_{\mathrm{rep}}^{(\alpha)}}d_A^2=2^2\times 32=|G|~.
\end{align}

\subsubsection{Linear representations}
Similarly, there are 16 one-dimensional representations, 12 two-dimensional representations, and 4 four-dimensional representations, thus:
\begin{align}
|C(G)|=N_{\mathrm{rep}}^{(0)}=16+12+4=32,
\end{align}
and again satisfies \eqref{repdim} by
\begin{align}
\sum_{A=1}^{N_{\mathrm{rep}}^{(\alpha)}}(d_A^{(0)})^2&=1^2\times 16+2^2\times 12+4^2\times 4=|G|~.
\end{align}

\subsection{Example of Minimal Order}
We consider an example with $|G|=64$, the smallest order for which $B(G)$ is nontrivial.
The group is expressed as a central extension
\begin{align}
    \Z^{(x)}_2\to G \to (\Z^{(a)}_2\times\Z^{(b)}_2)\ltimes \Z^{(c)}_8~.
\end{align}
where the generators of $(\Z_2\times\Z_2)\ltimes \Z_8$ is represented by $a,b,c$ with
\begin{align}
    a^2=b^2=c^8=1~,
\end{align}
with $a,b$ generating $\Z_2\times\Z_2$, $c$ generating $\Z_8$. The $\Z_2\times\Z_2$ action on $\Z_8$ is defined as
\begin{align}
    ac^\ell a = c^{3\ell}~, \quad bc^\ell b = c^{5\ell}~, \quad 0\le \ell\le 7~.
\end{align}
The generator of $\Z_2^{(x)}$ is given by $x\in Z(G)$ with $x^2=1$, satisfying
\begin{align}
    [a,b] = x~,
\end{align}
which characterizes the nontrivial central extension.
Generic group elements $ g\in G$ can be presented by a product
\begin{align}
g=a^{n}b^{m}c^{\ell}x^{k},
\end{align}
with $n,m,k=0,1$ and $\ell=0,1,\dots,7$, thus $|G|=64$. The multiplication law of $g,g'\in G$ is then defined as
\begin{align}
\begin{split}
gg'&=a^{n+n'}b^{m+m'}c^{3^{n'} 5^{m'}\ell + \ell'}x^{k+k'+n'm} \\
&= a^{n+n'}b^{m+m'}c^{\ell+\ell'+2\ell (n'+2 m')}x^{k+k'+n'm}~,
\end{split}
\end{align}
which completes the definition of $G$.
Let us describe the condition that $g,g'\in G$ commute with each other: 
\begin{align}
gg'=g'g\quad\Leftrightarrow\quad\ell n'=\ell' n\mod 4,\quad n'm=nm', \ell'm=\ell m'\mod2.\label{ggpgpg}
\end{align}

\subsubsection{Multiplier System}
Define a multiplier system on this group as:
\begin{align}
\alpha(g,g')\coloneqq(-1)^{m\ell'}.
\end{align}
This satisfies the cocycle condition, and all conjugacy classes are $\alpha$-regular due to the last equation in \eqref{ggpgpg}.

\subsubsection{Projective Representations}
This group has 16 two-dimensional irreducible representations and no one-dimensional representations, implying $\alpha$ represents a nontrivial element of $B(G)\cong \mathbb{Z}_2$. Hence,
\begin{align}
|C^{(\alpha)}(G)|=N_{\mathrm{rep}}^{(\alpha)}=16,
\end{align}
and the sum of squares becomes $|G|$:
\begin{align}
\sum_{A=1}^{N_{\mathrm{rep}}^{(\alpha)}}d_A^2=2^2\times 16=|G|.
\end{align}

\subsubsection{Linear Representations}
The linear representations consist of 8 one-dimensional, 6 two-dimensional, and 2 four-dimensional representations, thus:
\begin{align}
|C(G)|=N_{\mathrm{rep}}^{(0)}=8+6+2=16,
\end{align}
and the sum of squares becomes $|G|$:
\begin{align}
\sum_{A=1}^{N_{\mathrm{rep}}^{(\alpha)}}(d_A^{(0)})^2=1^2\times 8+2^2\times 6+4^2\times 2=|G|.
\end{align}

\section{Physical Applications of Bogomolov Multiplier}
\label{sec:physics}
\subsection{Review: Absence of String Order Parameters in (1+1)D SPT Phases}

{An SPT phase} cannot be characterized by local order parameters. Instead, in many cases a (1+1)D $G$ SPT phase is characterized through a non-local order parameter called a \textit{string order parameter}~\cite{Pollmann2010spectrum, Turner2011onedimensionalphase, PerezGarcia2008string, Haegeman2012orderparameter}. Suppose that {the} $G$ symmetry acts onsite. The idea is to consider a symmetry operator $U_g$ with $g\in G$ {supported on} an open region $\Sigma$ {of} length $l$, 
{with boundaries labeled $L$ and $R$.}
The string order parameter is defined as the ground state expectation value of the operator in the form of
\begin{align}
    \mathcal{S}_g(\Sigma,V_g^{(L)}, V_g^{(R)}) = \bra{\text{GS}} V_g^{(L)} U_g(\Sigma) V_g^{(R)}\ket{\text{GS}}~,
\end{align}
where $V_g^{(L)}, V_g^{(R)}$ are local unitaries supported at the ends of $\Sigma$. In general, there exists a choice of local operators $V_g^{(L)}, V_g^{(R)}$ {such that} the above expectation {value remains} nonzero in the limit of $l\to \infty.$ 

The operator $V_g^{(L)} U_g(\Sigma)V_g^{(R)}$ creates a pair of $g\in G$ symmetry defects when acting on a ground state.
In principle, there is freedom in the choice of terminations $V_g^{(L)}, V_g^{(R)}$, which leads to different choices of symmetry defects for a given symmetry operator $U_g$ on a closed 1d chain.  The above choice of $V_g^{(L)}, V_g^{(R)}${, which yields a} non-vanishing string order parameter, corresponds to a specific, preferred choice of {a} $g$ symmetry defect, such that the defect does not create excitations.

With the above choices of $V_g^{(L)}, V_g^{(R)}$, a large class of (1+1)D $G$ SPT phases {can be} characterized {by the} electric charge under $Z(g)$ carried by the symmetry defects $V_g^{(L)}, V_g^{(R)}$,
\begin{align}
    U_hV_g^{(L)} U_h^\dagger = \varphi_g(h) V_g^{(L)}~, \quad  U_hV_g^{(R)} U_h^\dagger = \varphi_g(h)^* V_g^{(R)}~.
\end{align}
When a (1+1)D SPT phase is characterized by $\alpha\in H^2(G,U(1))$, an electric charge carried by $V_g^{(L)}$ (or $V_g^{(R)\dagger}$) under $Z(g)$ is given by $\varphi_g \in \text{Rep}(Z(g))$ introduced in \eqref{eq:slant}; $\varphi_g(h) =\alpha(g,h)/\alpha(h,g)$~\cite{Pollmann2012detection}.

Let us consider an example of {a} string order parameter in {a} (1+1)D SPT phase with $\Z_2\times \Z_2$ symmetry. 
The model is given by a cluster state with a qubit at each site,
\begin{align}
    H_{\text{SPT}} = -\sum_{j} Z_{j-1}X_jZ_{j+1}~,
\end{align}
with $\Z_2\times\Z_2$ symmetry
\begin{align}
    \Z_2^{(\text{even})}: \prod_{j\in\text{even}}X_j, \quad \Z_2^{(\text{odd})}: \prod_{j\in\text{odd}}X_j~.
\end{align}
This SPT Hamiltonian is obtained by applying a finite-depth circuit called an SPT entangler, on a trivial symmetric product state;
\begin{align}
    H_{\text{SPT}} = U_{CZ} H_0 U^\dagger_{CZ}~,
\end{align}
with
\begin{align}
    H_{0} = -\sum_{j} X_j~, \quad U_{CZ} = \prod_{j} CZ_{j,j+1}~,
\end{align}
where $CZ$ is a Controlled-$Z$ operator.
The string order parameters of $H_{\text{SPT}}$ for the $\Z_2\times\Z_2$ symmetry are simply obtained by products of local stabilizer Hamiltonians, which take the form
\begin{align}
    \mathcal{S}_{\text{even}} = Z^{(L_{\text{odd}})}\left(\prod_{j\in\Sigma_{\text{even}}} X_j \right) Z^{(R_{\text{odd}})}~, \quad \mathcal{S}_{\text{odd}} = Z^{(L_{\text{even}})}\left(\prod_{j\in\Sigma_{\text{odd}}} X_j \right) Z^{(R_{\text{even}})}~.
\end{align}
These operators leave the ground state invariant.
The string order parameter $\mathcal{S}_{\text{even}}$ has the $Z$ operator at the odd site $L_{\text{odd}}$, and therefore carries electric charge under $\Z_2^{\text{odd}}$, and vice versa.

\subsubsection{(1+1)D SPT phase for Bogomolov multiplier}
When {a} (1+1)D $G$ SPT phase
is characterized by a cocycle $\alpha$ belonging to the Bogomolov multiplier $B(G)$, the string order parameter fails to detect it. This is because every $g\in G$ is $\alpha$-regular, implying that $\varphi_g$ is trivial; consequently, the symmetry defects $V_L, V_R$ do not carry an electric charge. The absence of {a} string order parameter was first pointed out in Ref.~\onlinecite{Pollmann2012detection} using the symmetry group given in Sec.~\ref{sec:pollmann-turner}.

We now explicitly construct a lattice model for the SPT phase with $\alpha \in B(G)$. We use the choice of $G$ discussed in Sec.~\ref{sec:pollmann-turner}.
Let us consider a 1d chain with $|G|$ dimensional local Hilbert space on each site. 
The local Hilbert space at each site is spanned by basis states $\ket{g}$ labeled by group elements $g\in G$.
We define the operators
\begin{align}
    \overrightarrow{X}_g\ket{h} = \ket{gh}~, \quad \overleftarrow{X}_g\ket{h} = \ket{hg^{-1}}~.
\end{align}
A trivial gapped phase with $G$ symmetry is realized by the Hamiltonian
\begin{align}
    H_0 = -\sum_{j}\sum_{g\in G} \overrightarrow{X}_g^{(j)}~,
    \label{eq:trivialHam}
\end{align}
with onsite $G$ symmetry
\begin{align}
    U_g = \prod_{j} \overrightarrow{X}_g^{(j)}~.
\end{align}

As in the case of $\Z_2\times\Z_2$ cluster state, the SPT phase is obtained by conjugating the trivial Hamiltonian $H_0$ by an SPT entangler. The SPT entangler for $\alpha\in B(G)$ is given by the unitary
\begin{align}
    U_{CZ} = \prod_j CZ^{j,j}_{a,b} CZ^{j,j+1}_{a,b}~,
\end{align}
where each unitary $CZ_{j,j}^{a,b}, CZ_{j,j+1}^{a,b}$ are Controlled-Z-type operators
\begin{align}
    CZ_{j,j}^{a,b}\ket{g_j} = (-1)^{n_a(g_j) n_b(g_j)}\ket{g_j}~, \quad CZ_{j,j+1}^{a,b}\ket{g_j,g_{j+1}} = (-1)^{n_a(g_j) n_b(g_{j+1})}\ket{g_j,g_{j+1}}~,
\end{align}
and 
\begin{align}
    Z^{(j)}_{a}\ket{g_j} = (-1)^{n_a(g_j)}\ket{g_j}~, \quad Z^{(j)}_{b}\ket{g_j} = (-1)^{n_b(g_j)}\ket{g_j}~,
\end{align}
where $a,b$ label the generators of $G$ defined in \eqref{eq:pollmann_def}, and $n_a(g), n_b(g)$ are mod 2 numbers that count the number of $a,b$ {appearing in the word representation of $g \in G$ using the generators} $a,b,c_1,c_2,x,y_1,y_2$ (see e.g., \eqref{eq:labelgroupelement_PT}).

The SPT Hamiltonian is then given by
\begin{align}
    H_{\text{SPT}} = U_{CZ} H_0 U_{CZ}^\dagger =  -\sum_j\sum_{g\in G}(Z_a^{(j-1)})^{n_b(g)}(Z_b^{(j)})^{n_a(g)}\overrightarrow{X}_g^{(j)} (Z_a^{(j)})^{n_b(g)}(Z_b^{(j+1)})^{n_a(g)}~,
    \label{eq:SPTHam}
\end{align}
which again has the onsite $G$ symmetry $U_g$. 
The string order parameter is obtained by conjugating the symmetry operator $U_g$ on a region $\Sigma$ with the SPT entangler $U_{CZ}$. 
This results in a string order parameter localized at the right end of 
$U_g(\Sigma)$. Namely, $U_{CZ}U_g(\Sigma) U_{CZ}^\dagger = V_g^{(L)} U_g(\Sigma) V_g^{(R)}$ with $V_g$ the operator at the ends of $\Sigma$. For instance, the operator at the right end is given by
\begin{align}
    V_{g}^{(R)} = (Z_a^{(R-1)})^{n_b(g)} (Z_b^{(R)})^{n_a(g)}~.
\end{align}
The electric charge of $V_{g}^{(R)}$ under $h\in Z(g)$ is given by $\alpha(g,h)/\alpha(h,g)$, with $\alpha(g,h)=(-1)^{n_a(g)n_b(h)}\in B(G)$ introduced in \eqref{eq:BG Pollmann}. Since $\alpha$ belongs to {the} Bogomolov multiplier, this electric charge is trivial. 
The string order parameters thus do not carry {any nontrivial electric charge} under the symmetry.

{Although the string order parameters carry only trivial charges}, $H_{\text{SPT}}$ realizes a nontrivial SPT phase with  protected edge modes. 
To demonstrate this explicitly, we introduce a boundary to the SPT Hamiltonian while preserving the $G$ symmetry,
\begin{align}
\begin{split}
    H^{\text{(open)}}_{\text{SPT}} = &-\sum_{2\le j\le L-1}\sum_{g\in G} (Z_a^{(j-1)})^{n_b(g)}(Z_b^{(j)})^{n_a(g)}\overrightarrow{X}_g^{(j)} (Z_a^{(j)})^{n_b(g)}(Z_b^{(j+1)})^{n_a(g)}~\\
    &- \sum_{\substack{g\in G \\ n_b(g)= 0}}(Z_b^{(1)})^{n_a(g)}\overrightarrow{X}_g^{(1)} (Z_b^{(2)})^{n_a(g)} - \sum_{\substack{g\in G \\ n_a(g)= 0}}
    (Z_a^{(L-1)})^{n_b(g)}\overrightarrow{X}_g^{(L)} (Z_a^{(L)})^{n_b(g)}~,
    \label{eq:SPTHam boundry}
    \end{split}
\end{align}
where 
the last two terms are edge Hamiltonians that commute with the bulk Hamiltonian.
We can find a set of boundary operators that commute with the Hamiltonian, {and thus} act within the ground state Hilbert space,
\begin{align}
\begin{split}
    \eta^{(L)}_g &:= (Z_a^{(1)})^{n_b(g)} \overleftarrow{X}_g^{(1)}(Z_b^{(2)})^{n_a(g)}~, \\
    \eta^{(R)}_g &:= (Z_a^{(L-1)})^{n_b(g)} \overleftarrow{X}_g^{(L)}(Z_b^{(L)})^{n_a(g)}~,
   \end{split}
\end{align}
where $\eta_g^{(L)}, \eta_g^{(R)}$ correspond to left and right edge modes.
These operators follow a projective representation of $G$
\begin{align}
    \eta^{(L)}_g \eta^{(L)}_h = \eta^{(L)}_{gh} \alpha(g,h)~, \quad \eta^{(R)}_g \eta^{(R)}_h = \eta^{(R)}_{gh} \alpha(g,h)~,
\end{align}
with $\alpha(g,h) = (-1)^{n_a(g)n_b(h)}\in B(G)$. 
This implies the presence of nontrivial edge modes {that carry a projective representation of $G$ associated with the cocycle $\alpha \in B(G)$.}
Indeed, the ground state Hilbert space of $H^{\text{(open)}}_{\text{SPT}}$ has a structure of $(\mathbb{C}^2)^{(L)}\otimes (\mathbb{C}^2)^{(R)}$, where $\eta_g^{(L/R)}$ acts on $(\mathbb{C}^2)^{(L/R)}$ by an irreducible projective representation, leading to 4-fold ground state degeneracy.

\subsection{Distinct $\text{Rep}(G)$ Symmetry Broken Phases in (1+1)D}
\label{sec:SSB}

\subsubsection{Maximally $\text{Rep}(G)$ broken phases: Lattice models}

Let us consider a (1+1)D $G$ SPT phase characterized by $\alpha\in H^2(G,U(1))$. We gauge the $G$ symmetry of the SPT phase, resulting in a gapped phase with $\text{Rep}(G)$ symmetry. 
If the cocycle $\alpha$ belongs to the Bogomolov multiplier $B(G)$, the $\text{Rep}(G)$ symmetry is maximally broken~\cite{kobayashi2025soft}; all $\text{Rep}(G)$ symmetry operators act faithfully on the ground state Hilbert space, and the ground states have maximal degeneracy among all possible gapped phases. In Sec.~\ref{sec:soft} we discuss the classification of $\text{Rep}(G)$ maximally broken phase using symmetry TQFT.

Let us describe explicit lattice models for distinct maximally $\text{Rep}(G)$ broken phases.
The models are obtained by gauging the $G$ symmetry of either the trivial or SPT Hamiltonian in \eqref{eq:trivialHam}, \eqref{eq:SPTHam}. The $G$ gauging is performed by a Kramers-Wannier (KW)-type transformation of operators~\cite{fechisin2023noninvertible}:
\begin{align}
    \overrightarrow{X}^{(j)}_g\mapsto  \overleftarrow{X}_g^{(j)} \overrightarrow{X}_g^{(j+1)}~, \quad Z_a^{(j)} Z_a^{(j+1)} \mapsto Z_a^{(j+1)}~,\quad Z_b^{(j)} Z_b^{(j+1)} \mapsto Z_b^{(j+1)}~.
\end{align}We 

{Applying this KW transformation to Eqs.~\eqref{eq:trivialHam} and \eqref{eq:SPTHam} yields the following Hamiltonians:}
\begin{align}
    H^{(0)}_{\text{SSB}} = -\sum_{j}\sum_{g\in G}\overleftarrow{X}_g^{(j)} \overrightarrow{X}_g^{(j+1)}~,
\end{align}

\begin{align}
\begin{split}
    H^{(\alpha)}_{\text{SSB}} &=     -\sum_{j}\sum_{g\in G}\overleftarrow{X}_g^{(j)}  (Z_a^{(j)})^{n_b(g)}(Z_b^{(j+1)})^{n_a(g)} \overrightarrow{X}_g^{(j+1)}~.
\end{split}
\end{align}
These models have a non-invertible $\text{Rep}(G)$ symmetry given by a matrix product operator (MPO) defined as follows.
Let $\Gamma\in \text{Rep}(G)$ {be} an irreducible representation of $G$, with dimension $d_\Gamma$. Then we define a local MPO with bond dimension $d_\Gamma$ as~\cite{fechisin2023noninvertible}
\begin{align}
    Z_{\Gamma} = \sum_{g\in G} \Gamma(g) \otimes |g\rangle \langle g|~,
\end{align}
where $\Gamma(g)$ is a representation matrix acting on the bond Hilbert space. See Fig.~\ref{fig:MPO}.

The corresponding symmetry operator is then given by the MPO formed by $Z_\Gamma$,
\begin{align}
    W_{\Gamma} = \text{tr}\left[\prod_{j}Z^{(j)}_\Gamma \right] = \sum_{\{g_j\}} \text{tr}(\Gamma(g_1\dots g_L))\ket{g_1\dots g_L}\bra{g_1\dots g_L}~,
\end{align}
where the last expression is for a periodic chain with length $L$. This operator indeed satisfies the $\text{Rep}(G)$ fusion rule,
\begin{align}
    W_\Gamma W_{\Gamma'} = W_{\Gamma\otimes \Gamma'} = N_{\Gamma,\Gamma'}^{\Gamma''} W_{\Gamma''}~,
\end{align}
where $N_{\Gamma,\Gamma'}^{\Gamma''}$ is the fusion multiplicity of irreducible representations. 

The symmetry operator commutes with the local Hamiltonian of $H^{(0)}_{\text{SSB}}$
\begin{align}
    W_\Gamma \overleftarrow{X}_g^{(j)} \overrightarrow{X}_g^{(j+1)} = \overleftarrow{X}_g^{(j)} \overrightarrow{X}_g^{(j+1)} W_\Gamma~,
\end{align} 
which can be seen by
\begin{align}
\begin{split}
    W_\Gamma \overleftarrow{X}_g^{(j)}\overrightarrow{X}_g^{(j+1)}\ket{h_1,\dots h_j,h_{j+1},\dots h_L} &= W_\Gamma \ket{h_1,\dots h_jg^{-1},gh_{j+1},\dots h_L} \\
    &=  \text{tr}(\Gamma(h_1\dots h_L)) \ket{h_1,\dots h_jg^{-1},gh_{j+1},\dots h_L} \\
    &=   \overleftarrow{X}_g^{(j)}\overrightarrow{X}_g^{(j+1)}W_\Gamma\ket{h_1,\dots h_j,h_{j+1},\dots h_L}~.
    \end{split}
\end{align}
The symmetry $W_\Gamma$ also commutes with local Hamiltonians of $H^{(\alpha)}_{\text{SSB}}$. 

\begin{figure}[t]
    \centering
    \includegraphics[width=0.7\textwidth]{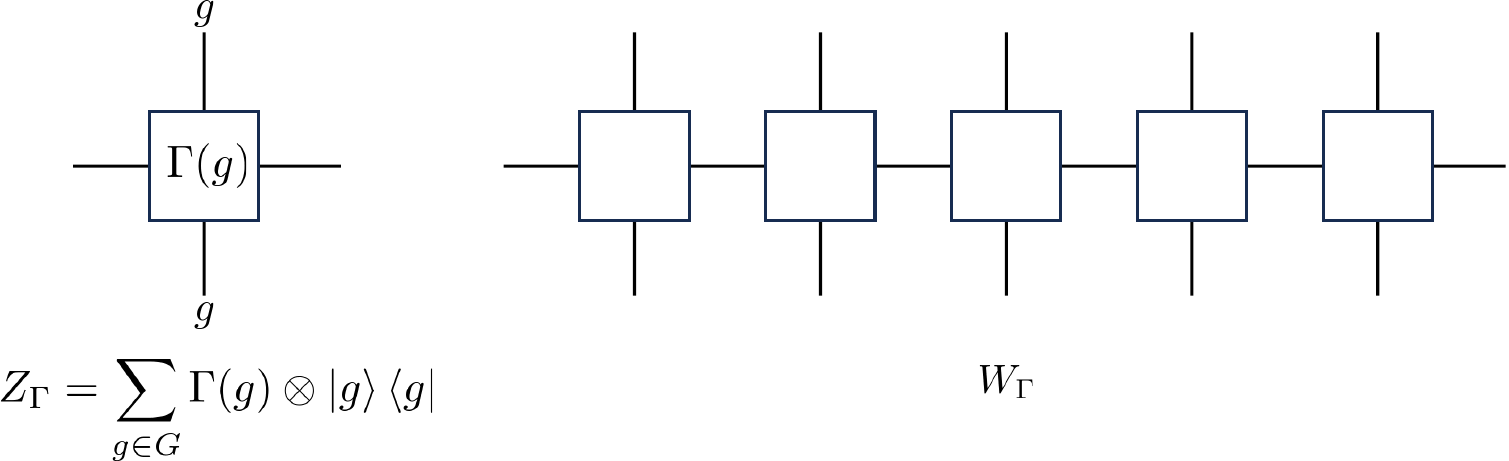}
    \caption{The non-invertible $\text{Rep}(G)$ symmetry is generated by an MPO $Z_\Gamma$, which gives a representation matrix $\Gamma(g)$ acting on the bond Hilbert space.}
    \label{fig:MPO}
\end{figure}

\subsubsection{Ground states and symmetry actions}
\label{sec:groundstates_SSB}

Let us study the ground states of the above $\text{Rep}(G)$ SSB phases and the action of the symmetry operator on them. 
{Among several possible approaches, we adopt the following method,} which can be readily generalized to more involved cases discussed later. We transform the SSB Hamiltonians into the short-range entangled (SRE) ones by a sequential circuit~\cite{Chen2024sequential}, which is a linear depth circuit of local unitaries. Since its depth {scales linearly with the} system size instead of constant, it can map long-range entangled SSB states into SRE states, and a local operator into a nonlocal one.

To describe the sequential circuit, let us introduce the following Controlled-$\overrightarrow{X}$ operator
\begin{align}
    (C\overrightarrow{X})_{i,j} \ket{g_i,g_j} = \ket{g_i,g_ig_j}~.
\end{align}
This operator acts on the local terms of the Hamiltonian $H_{\text{SSB}}^{(0)}$ as follows:
\begin{align}
\begin{split}
    (C\overrightarrow{X})_{i,j} \overrightarrow{X}_g^{(i)} (C\overrightarrow{X})^\dagger_{i,j} &= \overrightarrow{X}_g^{(i)}\overrightarrow{X}_g^{(j)}~,\\
    (C\overrightarrow{X})_{i,j} \overleftarrow{X}_g^{(i)}\overrightarrow{X}_g^{(j)}(C\overrightarrow{X})^\dagger_{i,j} &= \overleftarrow{X}_g^{(i)}~, \\
    (C\overrightarrow{X})_{i,j} \overleftarrow{X}_g^{(j)}(C\overrightarrow{X})^\dagger_{i,j} &= \overleftarrow{X}_g^{(j)}~.
\end{split}
\end{align}
In a periodic chain with length $L$, let us consider the following sequential circuit
\begin{align}
    U_{CX} =(C\overrightarrow{X})_{L-1, L}\dots (C\overrightarrow{X})_{2,3}(C\overrightarrow{X})_{1,2}~.
    \label{eq:UCX}
\end{align}

By using the transformation rule of local Hamiltonians, one can see that $H_{\text{SSB}}^{(0)}$ is transformed as
\begin{align}
    U_{CX} H_{\text{SSB}}^{(0)} U_{CX}^\dagger = -\sum_{g\in G}(\overleftarrow{X}^{(L)}_g\overrightarrow{X}^{(L)}_g)\prod_{j=1}^{L-1} \overrightarrow{X}^{(j)}_g - \sum_{j=1}^{L-1} \sum_{g\in G} \overleftarrow{X}^{(j)}_g~, 
\end{align}
where the local Hamiltonian $\overleftarrow{X}_g^{(L)} \overrightarrow{X}_g^{(1)}$ is transformed into the first nonlocal term, while the other Hamiltonians are transformed into onsite operators. Since the {transformed} Hamiltonian stabilizes a trivial product state in the bulk, the ground state after the sequential transformation becomes a product state in the form of
\begin{align}
    U_{CX}\ket{\text{GS}^{(0)}} = \bigotimes_{j=1}^{L-1} \left(\sum_{g_j\in G}\ket{g_j}\right) \otimes \ket{\text{GS}^{(0)}}_L~,
\end{align}
where $\ket{\text{GS}^{(0)}}_L$ is a local state at the site $L$ which minimizes the Hamiltonian $\sum_{g\in G}(\overleftarrow{X}^{(L)}_g\overrightarrow{X}^{(L)}_g)$. 
Such a local ground state is labeled by a conjugacy class $C\in C(G)$. For $C\in C(G)$, it is given by
\begin{align}
    \ket{\text{GS}^{(0)}(C)}_L = \sum_{g\in C} \ket{g}~.
\end{align}
Therefore the ground state degeneracy of $H_{\text{SSB}}^{(0)}$ is given by $|C(G)|=32$. We therefore label the ground states of $H^{(0)}_{\text{SSB}}$ by a conjugacy class $\ket{\text{GS}^{(0)}(C)}$.

The symmetry operator is mapped into an onsite operator at the site $L$,
\begin{align}
    U_{CX} W_\Gamma U_{CX}^\dagger = \sum_{g_L} \text{tr}(\Gamma(g_L))\ket{g_L}\bra{g_L}~.
\end{align}
Therefore, the symmetry action on each ground state is given by a group character
\begin{align}
    W_\Gamma \ket{\text{GS}^{(0)}(C)} = \chi_\Gamma(C)\ket{\text{GS}^{(0)}(C)}~.
\end{align}
This generalizes the symmetric cat state of SSB phases with finite invertible symmetries.

The same sequential circuit can be used to trivialize the other SSB Hamiltonian $H^{(\alpha)}_{\text{SSB}}$ associated with the nontrivial element of Bogomolov multiplier $\alpha\in B(G)$. The circuit transforms the local Hamiltonians in the following manner:

\begin{align}
\begin{split}
    (C\overrightarrow{X})_{i,j} (Z_b^{n_a(g)}\overrightarrow{X}_g)^{(i)} (C\overrightarrow{X})^\dagger_{i,j} &= (Z_b^{n_a(g)}\overrightarrow{X}_g)^{(i)}\overrightarrow{X}_g^{(j)}~,\\
    (C\overrightarrow{X})_{i,j} (\overleftarrow{X}_g Z_a^{n_b(g)})^{(i)}(Z_b^{n_a(g)}\overrightarrow{X}_g)^{(j)}(C\overrightarrow{X})^\dagger_{i,j} &= (Z_b^{n_a(g)}\overleftarrow{X}_gZ_a^{n_b(g)})^{(i)}(Z_b^{n_a(g)})^{(j)}~,\\
    (C\overrightarrow{X})_{i,j} (\overleftarrow{X}_g Z_a^{n_b(g)})^{(j)}(C\overrightarrow{X})^\dagger_{i,j} &=(Z_a^{n_b(g)})^{(i)}(\overleftarrow{X}_g Z_a^{n_b(g)})^{(j)}~.\\
\end{split}
\end{align}
By the sequential circuit $U_{CX}$ followed by a constant-depth circuit
\begin{align}
    U_{CZ} = \prod_{j=1}^{L-1}(CZ)^{a,b}_{j,j}
(CZ)^{a,b}_{j,j+1} ~,
\label{eq:UCZ}
\end{align}
the Hamiltonian $H_{\text{SSB}}^{(\alpha)}$ is transformed into
\begin{align}
    U_{CZ}U_{CX} H_{\text{SSB}}^{(\alpha)} U_{CX}^\dagger U_{CZ}^\dagger = -\sum_{g\in G}(\overleftarrow{X}_g\overrightarrow{X}_gZ_b^{n_a(g)}Z_a^{n_b(g)} )^{(L)}\prod_{j=1}^{L-1} \overrightarrow{X}^{(j)}_g - \sum_{j=1}^{L-1} \sum_{g\in G} \overleftarrow{X}^{(j)}_g~.
\end{align}
The bulk Hamiltonian again becomes onsite after the sequential transformation, so the ground state is again expressed as a product state
\begin{align}
    U_{CZ}U_{CX}\ket{\text{GS}^{(\alpha)}} = \bigotimes_{j=1}^{L-1} \left(\sum_{g_j\in G}\ket{g_j}\right) \otimes \ket{\text{GS}^{(\alpha)}}_L~,
\end{align}
where $\ket{\text{GS}^{(\alpha)}}_L$ is a ground state of a single site Hamiltonian $\sum_{g\in G}(\overleftarrow{X}_g\overrightarrow{X}_gZ_b^{n_a(g)}Z_a^{n_b(g)} )^{(L)}$.

To count the number of ground states $\ket{\text{GS}^{(\alpha)}}_L$, it is  instructive to regard the local Hamiltonian $\rho(g):= (\overleftarrow{X}_g\overrightarrow{X}_gZ_b^{n_a(g)}Z_a^{n_b(g)} )$ as a $|G|$-dimensional linear representation of $G$ acting on a single site Hilbert space. From this perspective, the number of ground states $d$ is the number of trivial representations contained in $\rho$ under decomposition into irreducible representations.
This can be computed using the orthogonality of characters,
\begin{align}
    \chi(g)= \text{tr}(\rho(g)) = \sum_{h\in G}\langle ghg^{-1}|h\rangle \frac{\alpha(g,h)}{\alpha(h,g)}=\sum_{h\in G}\delta_{gh,hg}\frac{\alpha(g,h)}{\alpha(h,g)}~,
    \label{eq:character_rho}
\end{align}
where $\alpha(g,h)= (-1)^{n_a(g)n_b(h)}\in B(G)$. 
The number of trivial representations $d$ is then simply given by the sum of characters
\begin{align}
\begin{split}
    d &= \frac{1}{|G|}\sum_{g\in G}\chi(g) = \sum_{g,h\in G}\delta_{gh,hg}\frac{\alpha(g,h)}{\alpha(h,g)} = |C^{(\alpha)}(G)|~,
    \label{eq:d1}
    \end{split}
\end{align}
where we applied Theorem \ref{thm:regclasscount}, and  $|C^{(\alpha)}(G)|$ is the number of $\alpha$-regular conjugacy class. Since $\alpha\in B(G)$ we have $|C^{(\alpha}(G)|=|C(G)|$, therefore the ground state degeneracy is again given by $|C(G)|=32$. \footnote{In general, if we gauge $G$ symmetry of a (1+1)D SPT phase with $\alpha\in H^2(G,U(1))$, the resulting theory carries ground state degeneracy given by $d=|C^{(\alpha)}(G)|$, where $|C^{(\alpha)}(G)|$ is the number of $\alpha$-regular conjugacy classes. This counts the number of charge neutral states in $G$ gauge theory.} 

In fact, there is again a direct correspondence between conjugacy classes $C\in C(G)$ and ground states. Let us take $C\in C(G)$, and pick a group element $g_1\in C$. Generic elements $g_j\in C$ with $1\le j\le |C|$ can be expressed as $g_j = k_j g_0 k_j^{-1}$, using some set of group elements $\{k_j\}$. Then the ground state is labeled by $C$ and given by
\begin{align}
    \ket{\text{GS}^{(\alpha)}(C)}_L = \sum_{j=1}^{|C|} \frac{\alpha(k_j,g_0)}{\alpha(g_0,k_j)}\ket{g_j}~. 
    \label{eq:state labeled by C}
\end{align}
Note that $k_j$ can be redefined as $k_j h$ with $h\in Z(g)$ and it shifts the phase factor by $\alpha(h,g_0)/\alpha(g_0,h)$, but since $\alpha\in B(G)$ the above state is independent of such choices.

The symmetry operator is again transformed into an onsite operator at the site $L$,
\begin{align}
    U_{CZ}U_{CX} W_\Gamma U_{CX}^\dagger U_{CZ}^{\dagger} =  \sum_{g_L} \text{tr}(\Gamma(g_L))\ket{g_L}\bra{g_L}~.
\end{align}
Therefore, the symmetry action on ground states {is} again given by
\begin{align}
    W_\Gamma \ket{\text{GS}^{(\alpha)}(C)} = \chi_\Gamma(C)\ket{\text{GS}^{(\alpha)}(C)}~.
\end{align}
This implies that two Hamiltonians $H_{\text{SSB}}^{(0)}$ and $H_{\text{SSB}}^{(\alpha)}$ have the same symmetry action on the low energy Hilbert space, therefore exhibit the same symmetry breaking pattern. In particular, the $\text{Rep}(G)$ symmetry is maximally broken in both phases. We summarize the properties of both SSB phases in Table \ref{tab:SSBHamiltonians}.

\begin{table}[htbp]
  \centering
        \begin{tabular}{c||c|c}
            & $H_{\text{SSB}}^{(0)}$ & $H_{\text{SSB}}^{(\alpha)}$  \\
          \hline
          & & \\
          $\text{Rep}(G)$ symmetry & $W_{\Gamma} = \text{tr}\left[\prod_{j}Z^{(j)}_\Gamma \right]$ &  $W_{\Gamma} = \text{tr}\left[\prod_{j}Z^{(j)}_\Gamma \right]$ \\
          & & \\
        \hline 
          Ground state degeneracy (periodic) & $|C(G)|=32$ & $|C^{(\alpha)}(G)|=32$ \\
          \hline
          & & \\
        Symmetry actions on ground states (periodic) & 
          $W_\Gamma \ket{\text{GS}^{(0)}(C)} = \chi_\Gamma(C)\ket{\text{GS}^{(0)}(C)}$ & $ W_\Gamma \ket{\text{GS}^{(\alpha)}(C)} = \chi_\Gamma(C)\ket{\text{GS}^{(\alpha)}(C)}$ \\
          & & \\
          \hline
          & & \\
          Local order parameters & $O^{(0)}_C = \sum_{g\in C} \overrightarrow{X}_g$ & $O^{(\alpha)}_C = \sum_{g\in C} (Z_b)^{n_a(g)}\overrightarrow{X}_g$ \\
          & & \\
          \hline
          & & \\
          Fusion rule of order parameters & $O^{(0)}_C\times O^{(0)}_{C'} = \sum_{C''\in C(G)} N^{C''}_{C,C'} O^{(0)}_{C''}$ & $O^{(\alpha)}_C\times O^{(\alpha)}_{C'} = \sum_{C''\in C(G)} N^{C''}_{C,C'} \alpha'(C,C') O^{(\alpha)}_{C''}$ \\
          & & \\
          \hline
          Ground state degeneracy (open) & $|G|=128$ & $|G|=128$ \\    
          \hline
          & & \\
          Symmetry actions on ground states (open) & 
          $W_\Gamma \ket{\text{GS}^{(0)}(g)} = \chi_\Gamma(g)\ket{\text{GS}^{(0)}(g)}$ & $ W_\Gamma \ket{\text{GS}^{(\alpha)}(g)} = \chi_\Gamma(g)\ket{\text{GS}^{(\alpha)}(g)}$ \\
          & & \\
          \end{tabular}
  \caption{The properties of the SSB Hamiltonians $H_{\text{SSB}}^{(0)},H_{\text{SSB}}^{(\alpha)}$. Ground state degeneracy and the symmetry actions on them are identical on both closed periodic and open chains. Meanwhile, the fusion rule of local order parameters distinguishes two SSB phases.}
\label{tab:SSBHamiltonians}
\end{table}

\subsubsection{Fusion rule of local order parameters}

The two Hamiltonians $H_{\text{SSB}}^{(0)}$ and $H_{\text{SSB}}^{(\alpha)}$
also support the same number of local order parameters.
The local order parameters of $H_{\text{SSB}}^{(0)}$ are onsite operators labeled by conjugacy classes $C\in C(G)$,
\begin{align}
    O^{(0)}_C = \sum_{g\in C} \overrightarrow{X}_g~,
    \label{local order parameter 0}
\end{align}
In contrast, the local order parameters of $H_{\text{SSB}}^{(\alpha)}$ {are} given by
\begin{align}
    O^{(\alpha)}_C = \sum_{g\in C} (Z_b)^{n_a(g)}\overrightarrow{X}_g~.
    \label{local order parameter alpha}
\end{align}
{In both cases, the operators commute with the respective Hamiltonians.}

The two SSB phases can be distinguished by the fusion rules of local order parameters under operator multiplication:
\begin{align}
    O^{(0)}_C\times O^{(0)}_{C'} = \sum_{C''\in C(G)} N^{C''}_{C,C'} O^{(0)}_{C''}~, 
    \label{fusion of local order parameters 0}
\end{align}
\begin{align}
    O^{(\alpha)}_C\times O^{(\alpha)}_{C'} = \sum_{C''\in C(G)} N^{C''}_{C,C'} \alpha'(C,C') O^{(\alpha)}_{C''}~, 
    \label{fusion of local order parameters alpha}
\end{align}
where $\alpha'(C,C')=\alpha'(g,h)=(-1)^{n_b(g)n_a(h)}$ for $C=[g],C'=[h]$, and $N^{C''}_{C,C'}$ is the fusion coefficient for multiplication of conjugacy classes~\cite{James_Liebeck_2001}. Since $[\alpha']=[\alpha]$ represents a nontrivial cohomology class in $H^2(G,U(1))$, the phase factor $\alpha'(C,C')$ in the fusion rule cannot be eliminated by shifting $O^{(\alpha)}(C)$ by phases. 

\subsubsection{Comment on boundaries}

Let us introduce boundaries to the Hamiltonians $H^{(0)}_{\text{SSB}}$ and $H^{(\alpha)}_{\text{SSB}}$. It is straightforward to  show that the boundary Hamiltonians have exactly the same ground state Hilbert space, with the same symmetry action; therefore, the boundaries cannot be used to distinguish the two phases.
The symmetric Hamiltonians with boundaries are given by
\begin{align}
    H^{(0,\text{open})}_{\text{SSB}} = -\sum_{1\le j\le L-1}\sum_{g\in G}\overleftarrow{X}_g^{(j)} \overrightarrow{X}_g^{(j+1)}~,
\end{align}
\begin{align}
\begin{split}
    H^{(\alpha,\text{open})}_{\text{SSB}} &= -\sum_{1\le j\le L-1}\sum_{g\in G}\overleftarrow{X}_g^{(j)}  (Z_a^{(j)})^{n_b(g)}(Z_b^{(j+1)})^{n_a(g)} \overrightarrow{X}_g^{(j+1)}~,
\end{split}
\end{align}
where the term between sites $L$ and $1$ is simply omitted. Using the sequential circuits \eqref{eq:UCX}, \eqref{eq:UCZ}, one can see that
\begin{align}
    U_{CX}H^{(0,\text{open})}_{\text{SSB}}U_{CX}^\dagger = U_{CZ}U_{CX}H^{(0,\text{open})}_{\text{SSB}}U_{CX}^\dagger U_{CZ}^\dagger = - \sum_{j=1}^{L-1} \sum_{g\in G} \overleftarrow{X}^{(j)}_g~.
\end{align}
Thus, these two Hamiltonians share the same ground state degeneracy. By the sequential circuit, the Hamiltonian becomes onsite and lacks terms acting on site $L$, therefore the ground state degeneracy of both SSB Hamiltonians {is} $|G|=128$, and each ground state is labeled by an element $g_L\in G$ that specifies the state $\ket{g_L}$ at the site $L$. The symmetry action on them is also the same,
\begin{align}
U_{CX}W_\Gamma U_{CX}^\dagger = U_{CZ}U_{CX}W_\Gamma U_{CX}^\dagger U_{CZ}^\dagger = \sum_{g_L\in G} \text{tr}(\Gamma(g_L))\ket{g_L}\bra{g_L}~,
\end{align}
which acts solely on the site $L$. 
{Hence, the boundaries cannot be used to distinguish these maximally symmetry-broken phases.
This observation motivates us to study gapped \textit{interfaces} between SSB phases, as discussed below.}

\subsection{Interface Modes Between Maximally  $\text{Rep}(G)$ Broken Phases}

Let us describe an interface between two distinct SSB phases $H_{\text{SSB}}^{(0)}$ and $H_{\text{SSB}}^{(\alpha)}$. 
We consider a closed chain with length $2L$, where the region $1\le j\le L$ has a Hamiltonian $H_{\text{SSB}}^{(0)}$, while $L+1\le j \le 2L$ has a Hamiltonian $H_{\text{SSB}}^{(\alpha)}$.  
This setup introduces interfaces at the boundaries $(2L,1)$ and $(L,L+1)$.
The Hamiltonian is given by
\begin{align}
    \begin{split}
    H^{(0|\alpha)}_{\text{SSB}} &= -\sum_{j=1}^{L-1}\sum_{g\in G}\overleftarrow{X}_g^{(j)} \overrightarrow{X}_g^{(j+1)}- \sum_{j=L+1}^{2L-1}\sum_{g\in G}\overleftarrow{X}_g^{(j)}  (Z_a^{(j)})^{n_b(g)}(Z_b^{(j+1)})^{n_a(g)} \overrightarrow{X}_g^{(j+1)} \\
    &- \sum_{g\in G} i^{n_a(g)n_b(g)}\overleftarrow{X}_g^{(2L)}  (Z_a^{(2L)})^{n_b(g)} \overrightarrow{X}_g^{(1)} - \sum_{g\in G}i^{n_a(g)n_b(g)}\overleftarrow{X}_g^{(L)} (Z_b^{(L+1)})^{n_a(g)} \overrightarrow{X}_g^{(L+1)}
    \end{split}
    \label{eq:interface}
\end{align}
where the last two terms are at the interfaces. The imaginary factors at the interfaces $i^{n_a(g)n_b(g)}$ are included to ensure Hermiticity of the Hamiltonian.
This Hamiltonian has the $\text{Rep}(G)$ symmetry generated by $W_\Gamma$. See Fig.~\ref{fig:interface}.

The local Hamiltonians at different sites commute with each other, but now the Hamiltonians at the interfaces get frustrated within themselves. In particular, writing the local Hamiltonian (up to a phase factor) as $\rho(g):=\overleftarrow{X}_g^{(L)} (Z_b^{(L+1)})^{n_a(g)} \overrightarrow{X}_g^{(L+1)}$ at the interface $(L,L+1)$, $\rho(g)$ becomes a \textit{projective} representation under $G$ rather than linear; $\rho(g)\rho(h) = \rho(gh) \alpha(g,h)$, with $\alpha=(-1)^{n_a(g)n_b(h)}\in B(G)$. This implies that these local Hamiltonians at the interface are intrinsically frustrated,
and one cannot minimize all local Hamiltonians at the same time. This is in contrast to the frustration free Hamiltonians $H_{\text{SSB}}^{(0)}$ or $H_{\text{SSB}}^{(\alpha)}$ without interfaces.

This is reminiscent of edge states of the 1d cluster state with $\Z_2\times\Z_2$ symmetry; a set of local operators at the boundary commuting with the Hamiltonian form a projective representation of $\Z_2\times\Z_2$ symmetry, which give rise to edge modes at the boundary. Indeed, we will later see that the algebra of interface Hamiltonians leads to nontrivial interface modes between distinct SSB phases.

One can again use a sequential circuit to bring the Hamiltonian into onsite form. We use the following circuits:
\begin{align}
    U_{CX}^{(0)} =(C\overrightarrow{X})_{L-1, L}\dots (C\overrightarrow{X})_{2,3}(C\overrightarrow{X})_{1,2}
\end{align}
\begin{align}
    U_{CX}^{(\alpha)} =(C\overrightarrow{X})_{2L-1, 2L}\dots (C\overrightarrow{X})_{L+2,L+3}(C\overrightarrow{X})_{L+1,L+2}
\end{align}
\begin{align}
    U^{(\alpha)}_{CZ} = \prod_{j=L+1}^{2L-1}(CZ)^{a,b}_{j,j}(CZ)^{a,b}_{j,j+1}
\end{align}
Then, one can show that $H^{(0|\alpha)}_{\text{SSB}}$ is transformed into
\begin{align}
\begin{split}
\tilde H^{(0|\alpha)}_{\text{SSB}} :=&
U^{(\alpha)}_{CZ}U_{CX}^{(\alpha)}U_{CX}^{(0)} H^{(0|\alpha)}_{\text{SSB}}U_{CX}^{(0)^\dagger}U_{CX}^{(\alpha)\dagger}U^{(\alpha)\dagger}_{CZ}  \\
=& -\sum_{j=1}^{L-1} \sum_{g\in G} \overleftarrow{X}^{(j)}_g   - \sum_{j=L+1}^{2L-1} \sum_{g\in G} \overleftarrow{X}^{(j)}_g \\
-& \sum_{g\in G}i^{n_a(g)n_b(g)}\overleftarrow{X}_g^{(2L)}  (Z_a^{(2L)})^{n_b(g)} \prod_{j=1}^{L} \overrightarrow{X}^{(j)}_g - \sum_{g\in G}i^{n_a(g)n_b(g)}
(Z_b^{(L)})^{n_a(g)} \overrightarrow{X}_g^{(L)}  \prod_{j=L}^{2L-1} \overrightarrow{X}^{(j)}_g
\end{split}
\end{align}
where the first two terms are onsite Hamiltonians in the bulk. The last two terms are nonlocal, which are obtained by the transformation of interface Hamiltonians.  See Fig.~\ref{fig:interface}.
This again allows us to write a ground state in the form of a product state,
\begin{align}
    \ket{\widetilde{\text{GS}}} = \bigotimes_{j=1}^{L-1} \left(\sum_{g_j\in G}\ket{g_j}\right) \bigotimes_{j=L+1}^{2L-1} \left(\sum_{g_j\in G}\ket{g_j}\right)\otimes 
    \ket{\widetilde{\text{GS}}}_{L,2L}~,
\end{align}
where $\ket{\widetilde{\text{GS}}}_{L,2L}$ is a two-body state supported on sites $L,2L$, and minimizes the two-body Hamiltonian
\begin{align}
    H'_{L,2L}= -\sum_{g\in G}i^{n_a(g)n_b(g)}\overrightarrow{X}_g^{(L)} (\overleftarrow{X}_g Z_a^{n_b(g)})^{(2L)} - \sum_{g\in G}i^{n_a(g)n_b(g)}\overleftarrow{X}_g^{(L)} (Z_b^{n_a(g)}\overrightarrow{X}_g )^{(2L)}
\end{align}
These two terms in $H'_{L,2L}$ encode the action of interface Hamiltonians on the ground states.
The ground state degeneracy is given by that of the above two-body Hamiltonian. 
It is convenient to further transform the above $H'_{L,2L}$ by the operator $U_{L,2L}=(CZ_{a,b})(C\overrightarrow{X})$ acting on the interfaces $L,2L$,
\begin{align}
\begin{split}
     H_{L,2L} &= U_{L,2L} H'_{L,2L} U^\dagger_{L,2L} \\
     &= -\sum_{g\in G}i^{n_a(g)n_b(g)} (\overrightarrow{X}_g Z_a^{n_b(g)})^{(L)} (\overleftarrow{X}_g \overrightarrow{X}_g Z_a^{n_b(g)}Z_b^{n_a(g)})^{(2L)} - \sum_{g\in G}i^{n_a(g)n_b(g)}(Z_b^{n_a(g)}\overleftarrow{X}_g )^{(L)}
    \end{split}
 \end{align}

To solve this Hamiltonian, we first focus on the site $2L$. The term $\rho(g) = (\overleftarrow{X}_g \overrightarrow{X}_g Z_a^{n_b(g)}Z_b^{n_a(g)})^{(2L)}$ forms a linear representation of $G$, and one can check that it is commutative; $\rho(g)\rho(h)=\rho(h)\rho(g)$. Indeed,
\begin{align}
\begin{split}
    \rho(g)\rho(h)\ket{k} &=(-1)^{n_b(g)n_a(k)+n_a(g)n_b(k)}
    (-1)^{n_b(h)n_a(k)+n_a(h)n_b(k)}\ket{ghkh^{-1}g^{-1}} \\
    &= (-1)^{n_b(gh)n_a(k)+n_a(gh)n_b(k)}\ket{ghkh^{-1}g^{-1}} \\
    &= (-1)^{n_b(hg)n_a(k)+n_a(hg)n_b(k)}\ket{hgkg^{-1}h^{-1}} \\
    &= \rho(h)\rho(g)\ket{k}
    \end{split}
\end{align}
where we used that $gh=hg x$ with some $x\in Z(G)$ (since $G$ is a central extension), and hence $ghkh^{-1}g^{-1}=hgkg^{-1}h^{-1}$.
Therefore $\rho$ decomposes into 1d representations $\{\pi\}$, and each 1d representation $\pi$ corresponds to an eigenstate $\ket{\Psi_\pi}$ of the Hamiltonian at $2L$. 

Let us fix one eigenstate $\ket{\Psi_\pi}$ at the site $2L$. The above {Hamiltonian $H_{L,2L}$ then reduces to a single-site Hamiltonian on site $L$,}
\begin{align}
\begin{split}
     H^{(\pi)}_{L,2L} &= -\sum_{g\in G}i^{n_a(g)n_b(g)}\pi(g) (\overrightarrow{X}_g Z_a^{n_b(g)})^{(L)} - \sum_{g\in G}i^{n_a(g)n_b(g)}(Z_b^{n_a(g)}\overleftarrow{X}_g )^{(L)}~,
    \end{split}
 \end{align}
which is labeled by a 1d $G$ representation $\pi$.
The ground state of $H_{L,2L}$ is then obtained by the Hamiltonian $H^{(\pi)}_{L,2L}$ with $\pi$ that yields the lowest energy.
By evaluating the above Hamiltonian on a computer, one can see that the ground state of the Hamiltonian $H^{(\pi)}_{L,2L}$ gives the smallest energy with the following four choices of the representation $\pi$ among 16 of them,
\begin{align}
    \pi(g) = 1, \quad \pi(g)= (-1)^{n_a(g)}, \quad \pi(g)=(-1)^{n_b(g)}, \quad \pi(g) = (-1)^{n_a(g)n_b(g)}~.
\end{align}
In these cases, the ground state of $H^{(\pi)}_{L,2L}$ is unique.

Therefore, the ground state Hilbert space has the form of
\begin{align}
    \mathcal{H} = \mathcal{H}_1 \oplus \mathcal{H}_{n_a} \oplus \mathcal{H}_{n_b} \oplus \mathcal{H}_{n_an_b}~,
\end{align}
where each $\mathcal{H}_\pi$ is the ground state Hilbert space that corresponds to the 1d representation $\pi$.
Accordingly, the ground state degeneracy $d$ of $H_{L,2L}$ is obtained by
\begin{align}
    d = d_1 + d_{n_a} + d_{n_b} + d_{n_an_b}~,
\end{align}
where $d_\pi$ is the number of 1d representations $\pi$ contained in $\rho(g)=(\overleftarrow{X}_g \overrightarrow{X}_g Z_a^{n_b(g)}Z_b^{n_a(g)})^{(2L)}$ upon irreducible decomposition, which counts the basis states $\ket{\Psi_\pi}$. This can be evaluated by using the orthogonality of characters. The character is computed in \eqref{eq:character_rho} as $\chi_\rho(g) = |Z(g)|$, then each $d_\pi$ is given by
\begin{align}
\begin{split}
    d_1 &= \frac{1}{|G|} \sum_{g\in G}|Z(g)|  = \sum_{g\in G}\frac{1}{|[g]|} = \sum_{C\in C(G)}1= |C(G)|=32~, \\
    d_{n_a} &= \frac{1}{|G|} \sum_{g\in G}(-1)^{n_a(g)}|Z(g)|
    =\sum_{g\in G}\frac{(-1)^{n_a(g)}}{|[g]|} = \sum_{C\in C(G)}(-1)^{n_a(C)} = 16~, \\
    d_{n_b} &=  \sum_{C\in C(G)}(-1)^{n_b(C)} = 4~, \\
    d_{n_an_b} &=  \sum_{C\in C(G)}(-1)^{n_a(C)}(-1)^{n_b(C)} = 4~, \\
    \end{split}
    \label{eq:computation of dpi}
\end{align}
where one can find the detailed descriptions of the conjugacy classes used for computations in Appendix \ref{app:conjugacyclasses}.
Therefore the ground state degeneracy of the interface Hamiltonian is given by
\begin{align}
    d=32+16+4+4=56>32~,
    \label{eq:interface_degeneracy}
\end{align}
demonstrating that the interface supports nontrivial modes.

Let us examine how the symmetry acts on the ground state Hilbert space with interfaces. Unlike the edge modes of SPT phases (or interfaces of SPT phases studied in e.g., Refs.~\onlinecite{seifnashri2024cluster, inamura202411dsptphasesfusion}), the symmetry action does not localize at the interfaces of SSB phases. Since the symmetry is spontaneously broken in the bulk, the bulk symmetry action is also nontrivial. Meanwhile, after performing a sequential circuit one can make the symmetry action localized at the interface.

By conjugating the symmetry operators $W_\Gamma$ by sequential circuits, we get
\begin{align}
    U_{L,2L}U^{(\alpha)}_{CZ}U_{CX}^{(\alpha)}U_{CX}^{(0)} W_\Gamma U_{CX}^{(0)^\dagger}U_{CX}^{(\alpha)\dagger}U^{(\alpha)\dagger}_{CZ} U^\dagger_{L,2L} = \chi_\Gamma(g_{2L})\ket{g_{2L}}\bra{g_{2L}}~.
\end{align}
Thus, by the action of the sequential circuit, the symmetry action becomes local at the site $2L$ and solely depends on the choice of $\ket{\Psi_\pi}$. Therefore the symmetry acts within each Hilbert space $\mathcal{H}_1, \mathcal{H}_{n_a}, \mathcal{H}_{n_b}, \mathcal{H}_{n_an_b}$.
Among the 56 ground states, if we focus on the $d_1=|C^{(\alpha)}(G)|=32$ ground states of $\mathcal{H}_1$, each state $\ket{\Psi_\pi(C)}$ is labeled by a conjugacy class $C\in C(G)$ and has the form of $\eqref{eq:state labeled by C}$.
The symmetry acts within this 32 dimensional Hilbert space in the same way as the case without interfaces, 
\begin{align}
    W_\Gamma \ket{\Psi_\pi(C)} = \chi_\Gamma(C)\ket{\Psi_\pi(C)}~.
    \label{eq:symaction_interface}
\end{align}
Since the conjugacy class $[h]$ of the state $\ket{h}$ at the site $2L$ is preserved under the actions of Hamiltonians $\rho(g)$, the states of $\mathcal{H}_{n_a}, \mathcal{H}_{n_b}, \mathcal{H}_{n_an_b}$ are also labeled by an element $C\in C(G)$ (although only a subset of $C(G)$ corresponds to a state of $\mathcal{H}_\pi$ with $\pi\neq 1$).
The symmetry actions on the sectors $\mathcal{H}_{n_a}, \mathcal{H}_{n_b}, \mathcal{H}_{n_an_b}$ are also diagonal in the basis $\ket{\Psi_\pi(C)}$, and take the form of \eqref{eq:symaction_interface}. 

\begin{figure}[t]
    \centering
    \includegraphics[width=0.9\textwidth]{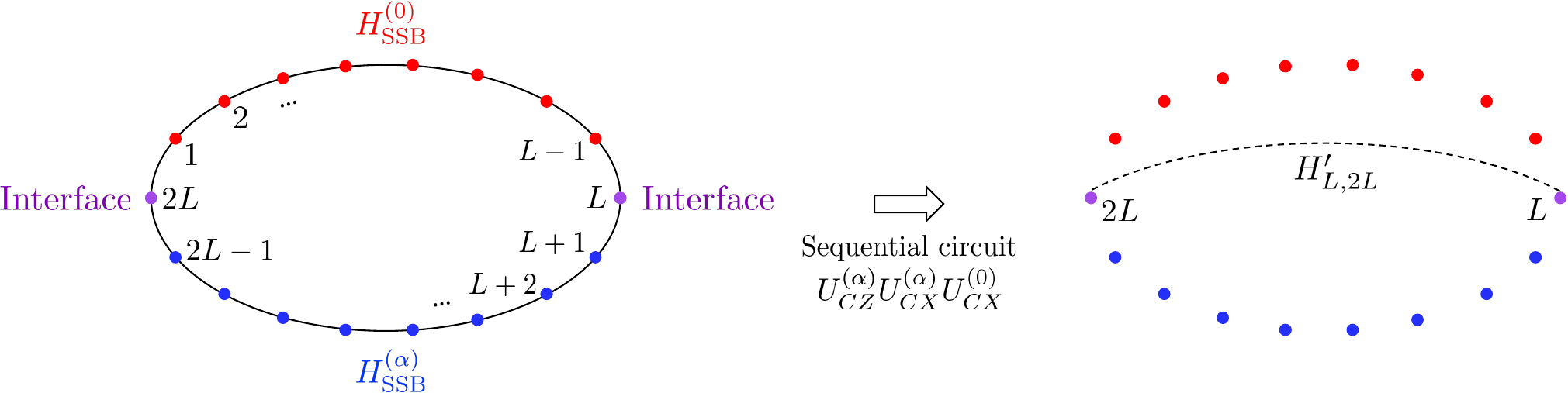}
    \caption{Left: The Hamiltonian with interfaces located sites $L$ and $2L$. Right: By a sequential circuit, the bulk becomes decoupled, leaving behind the two body Hamiltonian at the interfaces $L$ and $2L$. The symmetry action becomes localized at the interfaces after the action of the sequential circuit.}
    \label{fig:interface}
\end{figure}

\subsection{Soft Symmetry of (2+1)D Topological Order and Symmetry TQFT}
\label{sec:soft}

\subsubsection{Review: Soft symmetry of (2+1)D topological order}

In this section, we discuss the role of the Bogomolov multiplier in (2+1)D topological order. 
An element of group cohomology $\alpha\in H^2(G,U(1))$ generally corresponds to an ordinary (0-form) global symmetry of (2+1)D $G$ gauge theory, dubbed a \textit{gauged SPT} symmetry~\cite{Barkeshli2023codimension, Barkeshli:2022edm, Barkeshli:2023bta, Hsin:2019fhf, Yoshida2016generalized}. This can be understood as follows. 
Let us consider a (2+1)D trivial gapped phase with $G$ symmetry. Let us take a 1d submanifold of this 2d trivial phase, and place a (1+1)D $G$ SPT phase specified by $\alpha\in H^2(G,U(1))$ along this 1d locus. This can be done by acting the SPT entangler supported on the 1d submanifold. We then gauge the $G$ symmetry of the whole 2d system. After gauging the symmetry, the system becomes a (2+1)D topological order given by the $G$ gauge theory, and the decoration of (1+1)D SPT phase becomes an insertion of a symmetry defect in the $G$ gauge theory. This construction explains how the element of $H^2(G,U(1))$ defines a symmetry of $G$ gauge theory in (2+1)D. See Fig.~\ref{fig:gaugedSPT}.

The anyons in the $G$ gauge theory in (2+1)D are labeled by a pair $([g],\pi)$, where $[g]\in C(G)$ and $\pi\in \text{Rep}(Z(g))$ is an irreducible (linear) representation of $Z(g)$.
The first label corresponds to the magnetic flux, and the second corresponds to the electric charge of $Z(g)$ attached to it. The gauged SPT symmetry acts on the labels of anyons by
\begin{align}
    ([g],\pi)\to ([g],\pi\times \varphi_g)
    \label{eq:action of gauged SPT}
\end{align}
with $\varphi_g(h) = \alpha(g,h)/\alpha(h,g)$.
That is, the symmetry shifts the electric charge attached to the magnetic flux.

An example of such a gauged SPT symmetry is found in $\Z_2\times \Z_2$ gauge theory in (2+1)D, which is two copies of $\Z_2$ toric codes. Let us consider the nontrivial element $\alpha\in H^2(\Z_2\times\Z_2,U(1))$. The corresponding gauged SPT operator permutes the anyon as
\begin{align}
    m_1\to m_1e_2, \quad m_2\to m_2e_1, \quad e_1\to e_1, \quad e_2\to e_2~.
\end{align}

Let us now consider an element $\alpha\in B(G)$ of the Bogomolov multiplier. This leads to a global symmetry of (2+1)D $G$ gauge theory with the following exotic properties:
\begin{itemize}
    \item The symmetry does not permute labels of anyons~\cite{Davydov2014}. This is because every $g\in G$ is $\alpha$-regular, so $\varphi_g$ becomes trivial for any $g\in G$, meaning that the anyon permutation in \eqref{eq:action of gauged SPT} becomes trivial.    
    Therefore, the symmetry acts trivially on the ground state Hilbert space on a 2d torus.
    \item The symmetry fractionalization is also trivial, i.e., the anyons do not carry fractional charges under the symmetry.
    \item Nevertheless, the symmetry is still faithful; it acts on a ground state Hilbert space of higher genus surface by a nontrivial operator. While the symmetry do not permute the labels of anyons, the symmetry acts on the fusion vertices of the anyons, where a pair of anyons fuses into a single one at the junction.
\end{itemize}

The global symmetry of (2+1)D topological order with the above properties is referred to as a \textit{soft symmetry}~\cite{kobayashi2025soft}.

\begin{figure}[t]
    \centering
    \includegraphics[width=0.9\textwidth]{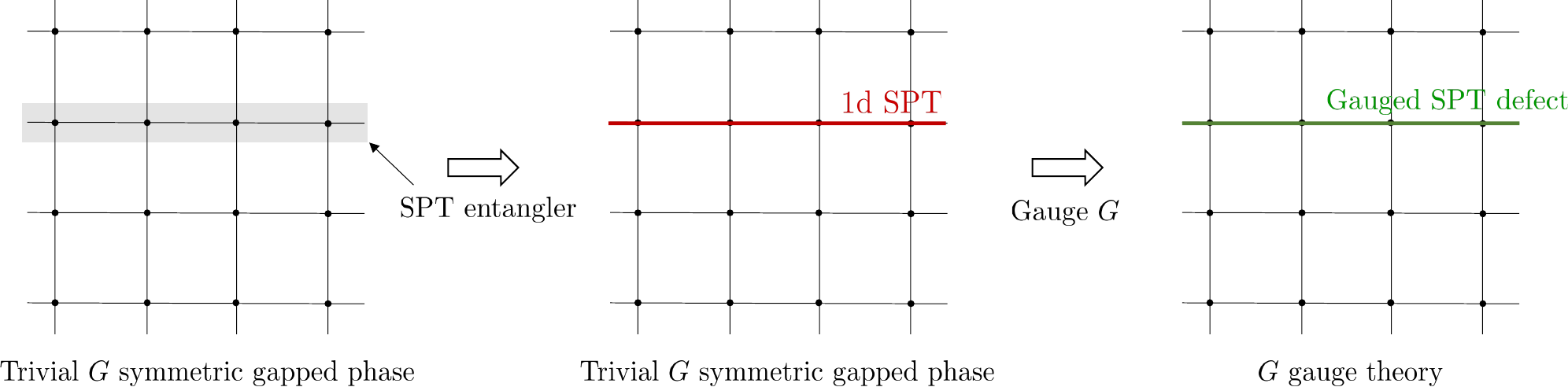}
    \caption{The construction of gauged SPT defect in (2+1)D $G$ gauge theory. We first start with a trivial gapped phase with $G$ symmetry (product state), and act the 1d SPT entangler along a 1d subsystem. This SPT entangler is labeled by group cohomology $\alpha\in H^2(G,U(1))$, and creates a 1d SPT phase at the subsystem. Then we gauge $G$ symmetry of the whole system to get $G$ gauge theory in (2+1)D. The decoration of 1d SPT phase results in the insertion of a symmetry defect in the $G$ gauge theory, which we call a gauged SPT defect.}
    \label{fig:gaugedSPT}
\end{figure}

\subsubsection{Review: Gapped boundary with identical condensed particles}

Let us review one physical consequence of soft symmetries. The gapped boundary condition of (2+1)D (untwisted) $G$ gauge theory is characterized by a pair $(K, \omega)$, where $K\subset G$ is a subgroup of $G$ and $\omega\in H^2(K, U(1))$~\cite{Beigi:2010htr}. Physically, $K\subset G$ specifies the unbroken gauge group at the boundary, and $\omega$ corresponds to a topological response localized at the boundary.

Now, suppose that the Bogomolov multiplier $B(G)$ is nontrivial, and let $\alpha\in B(G)$.
Let us consider a pair of gapped boundaries characterized by $(G,0)$ and $(G,\alpha)$. 
The boundary 
$(G,0)$ corresponds to the condition where magnetic fluxes are condensed. This corresponds to a Neumann boundary condition for the $G$ gauge field.

The boundary $(G,\alpha)$ corresponds to having the gauged SPT symmetry defect $\mathcal{D}_\alpha$ placed parallel to the gapped boundary $(G,0)$, and then ``pushing'' the gauged SPT defect onto the boundary. 
Since the gauged SPT defect $\mathcal{D}_\alpha$ corresponds to a soft symmetry, it does not change the label of the condensed anyons. Therefore, the new gapped boundary condition $(G,\alpha)$ again condenses the magnetic fluxes; the set of condensed particles for a pair of gapped boundaries $(G,0)$ and $(G,\alpha)$ are identical.

\subsubsection{Classification of (1+1)D Maximally Broken 
$\text{Rep}(G)$ Phases by the Bogomolov Multiplier}

In Sec.~\ref{sec:SSB} we discussed an example of a $\text{Rep}(G)$ SSB phase in (1+1)D obtained by gauging (1+1)D SPT phase characterized by Bogomolov multiplier.

To understand the general relation between maximally symmetry broken phases and Bogomolov multiplier, it is instructive to describe this $\text{Rep}(G)$ broken phase within the framework of symmetry TQFT. That is, generic (1+1)D phase with finite non-invertible symmetry can be described by a thin interval of (2+1)D TQFT sandwiched by a pair of boundary conditions. One boundary is gapped and called a symmetry boundary condition, where the topological operators at the boundary are regarded as a global symmetry of the system. The other boundary is called a dynamical boundary condition that encodes the choice of the (1+1)D theory. See Fig.~\ref{fig:symtft} (a).
In the case of $\text{Rep}(G)$ symmetry, the bulk TQFT is the (untwisted) $G$ gauge theory in (2+1)D. The symmetry boundary condition is given by the one condensing magnetic fluxes, i.e., the Neumann boundary condition for the $G$ gauge field. The ground state degeneracy of the (1+1)D gapped phase is given by counting an overlap of condensed particles in both gapped boundaries.
In the maximally broken phase, the dynamical boundary condition also condenses the magnetic particles, so that the ground state degeneracy is maximized. 
Such dynamical boundary conditions are labeled by a pair $(G,\alpha)$, with $\alpha$ an element of Bogomolov multiplier. See Fig.~\ref{fig:symtft} (b).

Therefore, we conclude that the (1+1)D gapped phase with maximally broken $\text{Rep}(G)$ symmetry is classified by the Bogomolov multiplier $B(G)$. 
Such phases are obtained by gauging the $G$ symmetry of SPT phases characterized by $\alpha\in B(G)$.

\begin{figure}[t]
    \centering
    \includegraphics[width=0.8\textwidth]{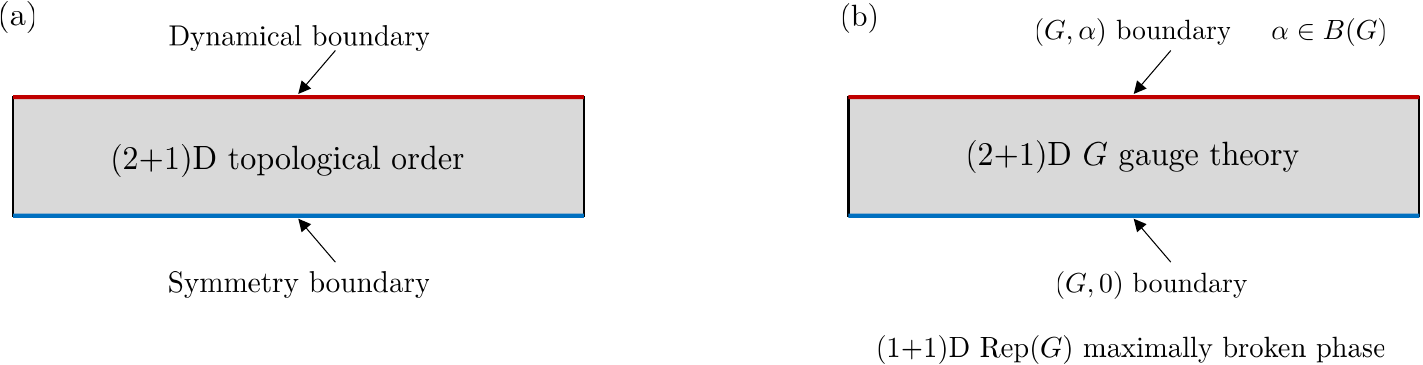}
    \caption{(a): A generic (1+1)D system with a non-invertible symmetry is obtained by an interval of (2+1)D topological order sandwiched by a pair of boundary conditions. The symmetry boundary condition is gapped, and the topological operators at the symmetry boundary define the symmetries of the interval and the resulting (1+1)D system. (b): Symmetry TQFT description for a (1+1)D gapped phase with maximally broken $\text{Rep}(G)$ symmetry. }
    \label{fig:symtft}
\end{figure}

\subsubsection{Interface modes between SSB phases from symmetry TQFT}

In Sec.~\ref{sec:SSB}, we discussed the interface mode between two distinct  $\text{Rep}(G)$ SSB phases in (1+1)D.

The ground state degeneracy of the interface Hamiltonian \eqref{eq:interface} can be derived within the framework of  the symmetry TQFT at the continuum level. 
In symmetry TQFT, the interface between distinct SSB phases is understood as putting an interval of a gauged SPT defect for $\alpha(g,h)=(-1)^{n_a(g)n_b(h)}\in B(G)$ 
{placed near the Neumann boundary of} $G$ gauge theory, see Fig.~\ref{fig:interface_symtft}. This is a condensation defect obtained by 1-gauging the $\Z_2\times \Z_2$ Wilson lines for 1d representations $W_{n_a}, W_{n_b}$ in the bulk~\cite{Roumpedakis:2022aik}, where $W_\pi$ is an Wilson line carrying the $G$ representation $\pi$.\footnote{This 1-gauging corresponds to the twisted 1-gauging in the presence of the nontrivial discrete torsion $\eta\in H^2(\Z_2\times\Z_2,U(1))$ of $\Z_2\times\Z_2$ symmetry~\cite{Barkeshli2023codimension}. }
Therefore by shrinking a size of the interval, this gauged SPT defect becomes a non-simple line operator $1 + W_{n_a} + W_{n_b} + W_{n_an_b}$ fused at the Neumann boundary (see the right of Fig.~\ref{fig:interface_symtft}). 
The corresponding ground state degeneracy is computed via the overlap of boundary states
\begin{align}
    d = \langle \mathcal{L}| \mathcal{L}\times (1+W_{n_a}+W_{n_b}+W_{n_an_b}) \rangle~,
\end{align}
where $\mathcal{L}$ is the Lagrangian algebra anyon for the Neumann boundary condition that characterizes the boundary state on a torus~\cite{Kaidi2022higher},
\begin{align}
    \mathcal{L} = \bigoplus_{C\in C(G)} (C,1)~,
\end{align}
where $(C,1)$ denotes a magnetic flux in (2+1)D $G$ gauge theory carrying $C\in C(G)$.
By using the invariance of $\ket{\mathcal{L}}$ under modular $S$ transformations~\cite{Lan2015gapped, Kaidi2022higher},
\begin{align}
    d =  \langle \mathcal{L}| M_{1} + M_{n_a} + M_{n_b} + M_{n_an_b}| \mathcal{L} \rangle~,
\end{align}
where $M_{\pi}$ is an operator that links the Wilson line for 1d representation $\pi$ to the state.

We then obtain
\begin{align}
\begin{split}
    \langle \mathcal{L}|M_1|\mathcal{L}\rangle &= \sum_{C\in C(G)} 1 = d_1~, \\
    \langle \mathcal{L}|M_{n_a}|\mathcal{L}\rangle &= \sum_{C\in C(G)} (-1)^{n_a(C)} = d_{n_a}~, \\
    \langle \mathcal{L}|M_{n_b}|\mathcal{L}\rangle &= \sum_{C\in C(G)} (-1)^{n_b(C)} = d_{n_b}~, \\
    \langle \mathcal{L}|M_{n_an_b}|\mathcal{L}\rangle &= \sum_{C\in C(G)} (-1)^{n_a(C)}(-1)^{n_b(C)} = d_{n_an_b}~, \\
    \end{split}
\end{align}
where $d_\pi$ are integers described in \eqref{eq:computation of dpi}.
This matches the ground state degeneracy  $d= d_1 + d_{n_a} + d_{n_b} + d_{n_an_b}=56$ in Eq.~\eqref{eq:interface_degeneracy}.

\begin{figure}[t]
    \centering
    \includegraphics[width=0.9\textwidth]{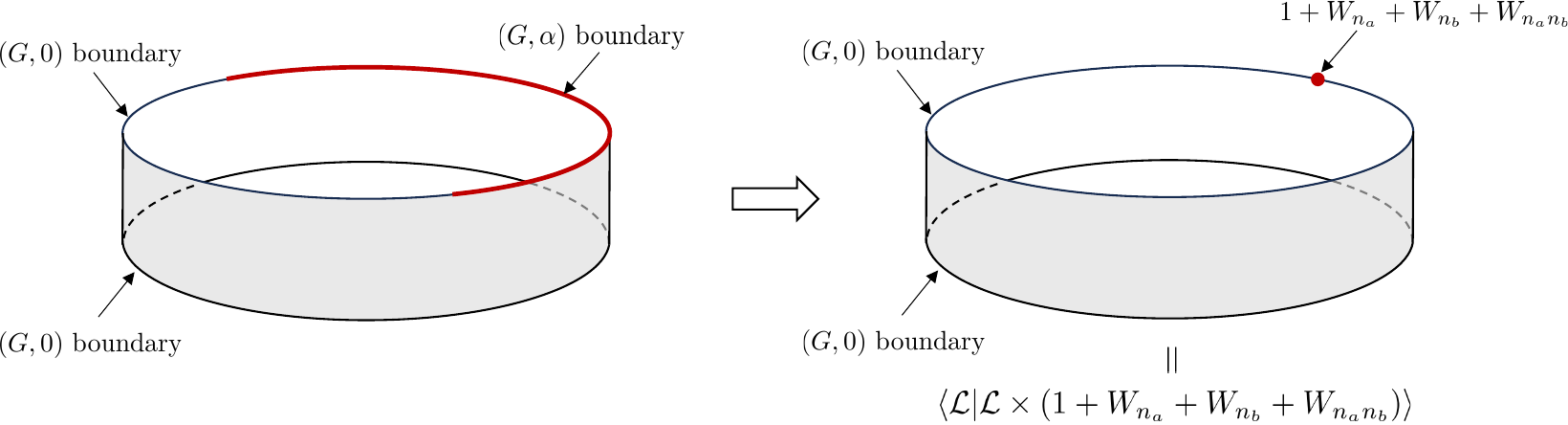}
    \caption{Symmetry TQFT picture for the interfaces between two distinct $(1+1)$D gapped phases with maximally broken $\text{Rep}(G)$ symmetry. On the dynamical (top) boundary, we have an interval of $(G,\alpha)$ boundary condition of (2+1)D $G$ gauge theory with $\alpha\in B(G)$. This is regarded as having a gauged SPT defect $\alpha$ along the interval of the top boundary (red curve). One can shrink the size of the interval for the gauged SPT defect, then we end up with a non-simple anyon $1+W_{n_a}+W_{n_b}+W_{n_an_b}$. The partition function is then expressed as the overlap of boundary states $\langle \mathcal{L}| \mathcal{L}\times (1+W_{n_a}+W_{n_b}+W_{n_an_b}) \rangle$.}
    \label{fig:interface_symtft}
\end{figure}

\subsubsection{Fusion rule of local order parameters}
In this symmetry TQFT picture, a local order parameter is a short topological line operator expanded along the thin interval of the strip, which becomes a local operator by shrinking the width of the interval. This topological operator has to end (i.e., condense) at a pair of gapped boundaries in Fig.~\ref{fig:symtft} (a) or (b), therefore given by lines of magnetic fluxes labeled by $([g],1)$ with $[g]\in C(G)$ for both SSB phases. 
These operators correspond to the local order parameters in the lattice model $O_C^{(0)}$ or $O_C^{(\alpha)}$ discussed in 
\eqref{local order parameter 0}, \eqref{local order parameter alpha}.

Then, the fusion rule of local order parameters on each SSB phase corresponds to the fusion of these short line operators ending at the interval. In our setup, the fusion of magnetic fluxes $[g],[h]$ into $[gh]$ with $g,h\in G$ acquires a phase $\alpha(g,h)$ at the gapped boundary $(G,\alpha)$. We note that in this example, such phase factor corresponds to a multiplication morphism of Lagrangian algebra at the gapped boundary $(G,\alpha)$; the structure of the Lagrangian algebra for $(G,\alpha)$ boundary differs from that of $(G,0)$ by the multiplication morphism, while the condensed anyons are identical \cite{kobayashi2025soft}. This leads to the distinct fusion rules of the local order parameters \eqref{fusion of local order parameters 0}, \eqref{fusion of local order parameters alpha}.


\section*{Acknowledgments}

R.K. thanks Maissam Barkeshli for related collaborations. H.W. thanks discussions with Takuma Saito.
R.K. and H.W. thank Okinawa Institute for Science and Technology for hosting
the program “Generalized Symmetries in Quantum Matter” in 2025, during which part of the work
is completed.

\section*{Funding Declaration}

R.K. is supported by U.S. Department of Energy through grant number DE-SC0009988 and the Sivian Fund. H.W. is supported by JSPS KAKENHI Grant No.~JP24K00541.

\bibliographystyle{utphys}
\bibliography{biblio}

\providecommand{\href}[2]{#2}\begingroup\raggedright\begin{thebibliography}{10}

\bibitem{PhysRevB.83.035107}
X.~Chen, Z.-C. Gu, and X.-G. Wen, ``Classification of gapped symmetric phases in one-dimensional spin systems,'' \href{http://dx.doi.org/10.1103/PhysRevB.83.035107}{{\em Phys. Rev. B} {\bfseries 83} (Jan, 2011) 035107}. \url{https://link.aps.org/doi/10.1103/PhysRevB.83.035107}.

\bibitem{PhysRevB.85.075125}
F.~Pollmann, E.~Berg, A.~M. Turner, and M.~Oshikawa, ``Symmetry protection of topological phases in one-dimensional quantum spin systems,'' \href{http://dx.doi.org/10.1103/PhysRevB.85.075125}{{\em Phys. Rev. B} {\bfseries 85} (Feb, 2012) 075125}. \url{https://link.aps.org/doi/10.1103/PhysRevB.85.075125}.

\bibitem{PhysRevLett.117.096404}
H.~Watanabe, H.~C. Po, M.~P. Zaletel, and A.~Vishwanath, ``Filling-enforced gaplessness in band structures of the 230 space groups,'' \href{http://dx.doi.org/10.1103/PhysRevLett.117.096404}{{\em Phys. Rev. Lett.} {\bfseries 117} (Aug, 2016) 096404}. \url{https://link.aps.org/doi/10.1103/PhysRevLett.117.096404}.

\bibitem{PhysRevX.7.041069}
J.~Kruthoff, J.~de~Boer, J.~van Wezel, C.~L. Kane, and R.-J. Slager, ``Topological classification of crystalline insulators through band structure combinatorics,'' \href{http://dx.doi.org/10.1103/PhysRevX.7.041069}{{\em Phys. Rev. X} {\bfseries 7} (Dec, 2017) 041069}. \url{https://link.aps.org/doi/10.1103/PhysRevX.7.041069}.

\bibitem{SI}
H.~C. Po, A.~Vishwanath, and H.~Watanabe, ``Symmetry-based indicators of band topology in the 230 space groups,'' \href{http://dx.doi.org/10.1038/s41467-017-00133-2}{{\em Nature Communications} {\bfseries 8} no.~1, (2017) 50}. \url{https://doi.org/10.1038/s41467-017-00133-2}.

\bibitem{TQC}
B.~Bradlyn, L.~Elcoro, J.~Cano, M.~G. Vergniory, Z.~Wang, C.~Felser, M.~I. Aroyo, and B.~A. Bernevig, ``Topological quantum chemistry,'' \href{http://dx.doi.org/10.1038/nature23268}{{\em Nature} {\bfseries 547} no.~7663, (2017) 298--305}. \url{https://doi.org/10.1038/nature23268}.

\bibitem{PhysRevLett.119.127202}
H.~C. Po, H.~Watanabe, C.-M. Jian, and M.~P. Zaletel, ``Lattice homotopy constraints on phases of quantum magnets,'' \href{http://dx.doi.org/10.1103/PhysRevLett.119.127202}{{\em Phys. Rev. Lett.} {\bfseries 119} (Sep, 2017) 127202}. \url{https://link.aps.org/doi/10.1103/PhysRevLett.119.127202}.

\bibitem{OgataTasaki}
Y.~Ogata and H.~Tasaki, ``Lieb--schultz--mattis type theorems for quantum spin chains without continuous symmetry,'' \href{http://dx.doi.org/10.1007/s00220-019-03343-5}{{\em Communications in Mathematical Physics} {\bfseries 372} no.~3, (2019) 951--962}. \url{https://doi.org/10.1007/s00220-019-03343-5}.

\bibitem{PhysRevX.8.011040}
R.~Thorngren and D.~V. Else, ``Gauging spatial symmetries and the classification of topological crystalline phases,'' \href{http://dx.doi.org/10.1103/PhysRevX.8.011040}{{\em Phys. Rev. X} {\bfseries 8} (Mar, 2018) 011040}. \url{https://link.aps.org/doi/10.1103/PhysRevX.8.011040}.

\bibitem{Cheng2016translational}
M.~Cheng, M.~Zaletel, M.~Barkeshli, A.~Vishwanath, and P.~Bonderson, ``Translational symmetry and microscopic constraints on symmetry-enriched topological phases: A view from the surface,'' \href{http://dx.doi.org/10.1103/physrevx.6.041068}{{\em Physical Review X} {\bfseries 6} no.~4, (Dec., 2016) }. \url{http://dx.doi.org/10.1103/PhysRevX.6.041068}.

\bibitem{Cheng2023LSM}
M.~Cheng and N.~Seiberg, ``Lieb-schultz-mattis, luttinger, and ’t hooft - anomaly matching in lattice systems,'' \href{http://dx.doi.org/10.21468/scipostphys.15.2.051}{{\em SciPost Physics} {\bfseries 15} no.~2, (Aug., 2023) }. \url{http://dx.doi.org/10.21468/SciPostPhys.15.2.051}.

\bibitem{Tinkham}
M.~Tinkham, {\em Group theory and quantum mechanics}.
\newblock Courier Corporation, 2003.

\bibitem{bradley2009mathematical}
C.~Bradley and A.~Cracknell, {\em The mathematical theory of symmetry in solids: representation theory for point groups and space groups}.
\newblock Oxford University Press, 2009.

\bibitem{Curtis}
C.~W. Curtis and I.~Reiner, {\em Representation theory of finite groups and associative algebras}.
\newblock American Mathematical Soc., 1966.

\bibitem{Haggarty}
R.~Haggarty and J.~Humphreys, ``Projective characters of finite groups,'' {\em Proceedings of the London Mathematical Society} {\bfseries 3} no.~1, (1978) 176--192.

\bibitem{Karpilovsky}
G.~Karpilovsky, {\em Projective representations of finite groups}.
\newblock Marcel Dekker, 1988.

\bibitem{Bogomolov}
F.~A. Bogomolov, ``The brauer group of quotient spaces of linear representations,'' \href{http://dx.doi.org/10.1070/IM1988v030n03ABEH001024}{{\em Mathematics of the USSR-Izvestiya} {\bfseries 30} no.~3, (Jun, 1988) 455}. \url{https://dx.doi.org/10.1070/IM1988v030n03ABEH001024}.

\bibitem{Moravec}
P.~Moravec, ``Unramified brauer groups of finite and infinite groups,'' {\em American Journal of Mathematics} {\bfseries 134} no.~6, (2012) 1679--1704. \url{http://www.jstor.org/stable/23358379}.

\bibitem{Davydov2014}
A.~Davydov, ``Bogomolov multiplier, double class-preserving automorphisms, and modular invariants for orbifolds,'' \href{http://dx.doi.org/10.1063/1.4895764}{{\em Journal of Mathematical Physics} {\bfseries 55} no.~9, (2014) }, \href{http://arxiv.org/abs/1312.7466}{{\ttfamily arXiv:1312.7466 [math]}}. \url{http://dx.doi.org/10.1063/1.4895764}.

\bibitem{Pollmann2012detection}
F.~Pollmann and A.~M. Turner, ``Detection of symmetry-protected topological phases in one dimension,'' \href{http://dx.doi.org/10.1103/physrevb.86.125441}{{\em Physical Review B} {\bfseries 86} no.~12, (Sept., 2012) }. \url{http://dx.doi.org/10.1103/PhysRevB.86.125441}.

\bibitem{McGreevy:2022oyu}
J.~McGreevy, ``Generalized symmetries in condensed matter,'' \href{http://dx.doi.org/10.1146/annurev-conmatphys-040721-021029}{{\em Annual Review of Condensed Matter Physics} {\bfseries 14} no.~1, (Mar., 2023) 57--82}, \href{http://arxiv.org/abs/2204.03045}{{\ttfamily arXiv:2204.03045 [cond-mat.str-el]}}. \url{http://dx.doi.org/10.1146/annurev-conmatphys-040721-021029}.

\bibitem{schafernameki2023ictp}
S.~Schafer-Nameki, ``Ictp lectures on (non-)invertible generalized symmetries,'' \href{http://arxiv.org/abs/2305.18296}{{\ttfamily arXiv:2305.18296 [hep-th]}}.

\bibitem{shao2024tasi}
S.-H. Shao, ``What's done cannot be undone: Tasi lectures on non-invertible symmetries,'' \href{http://arxiv.org/abs/2308.00747}{{\ttfamily arXiv:2308.00747 [hep-th]}}.

\bibitem{brennan2023introduction}
T.~D. Brennan and S.~Hong, ``Introduction to generalized global symmetries in qft and particle physics,'' 2023.
\newblock \url{https://arxiv.org/abs/2306.00912}.

\bibitem{bhardwaj2023lectures}
L.~Bhardwaj, L.~E. Bottini, L.~Fraser-Taliente, L.~Gladden, D.~S.~W. Gould, A.~Platschorre, and H.~Tillim, ``Lectures on generalized symmetries,'' 2023.
\newblock \url{https://arxiv.org/abs/2307.07547}.

\bibitem{kobayashi2025soft}
R.~Kobayashi and M.~Barkeshli, ``Soft symmetries of topological orders,'' 2025.
\newblock \url{https://arxiv.org/abs/2501.03314}.

\bibitem{Ji:2019jhk}
W.~Ji and X.-G. Wen, ``{Categorical symmetry and noninvertible anomaly in symmetry-breaking and topological phase transitions},'' \href{http://dx.doi.org/10.1103/PhysRevResearch.2.033417}{{\em Phys. Rev. Res.} {\bfseries 2} no.~3, (2020) 033417}, \href{http://arxiv.org/abs/1912.13492}{{\ttfamily arXiv:1912.13492 [cond-mat.str-el]}}.

\bibitem{Kaidi:2022cpf}
J.~Kaidi, K.~Ohmori, and Y.~Zheng, ``{Symmetry TFTs for Non-Invertible Defects},''  (9, 2022) , \href{http://arxiv.org/abs/2209.11062}{{\ttfamily arXiv:2209.11062 [hep-th]}}.

\bibitem{Freed:2022qnc}
D.~S. Freed, G.~W. Moore, and C.~Teleman, ``{Topological symmetry in quantum field theory},''  (9, 2022) , \href{http://arxiv.org/abs/2209.07471}{{\ttfamily arXiv:2209.07471 [hep-th]}}.

\bibitem{Thorngren:2019iar}
R.~Thorngren and Y.~Wang, ``{Fusion Category Symmetry I: Anomaly In-Flow and Gapped Phases},'' \href{http://arxiv.org/abs/1912.02817}{{\ttfamily arXiv:1912.02817 [hep-th]}}.

\bibitem{Lichtman2021}
T.~Lichtman, R.~Thorngren, N.~H. Lindner, A.~Stern, and E.~Berg, ``Bulk anyons as edge symmetries: Boundary phase diagrams of topologically ordered states,'' \href{http://dx.doi.org/10.1103/physrevb.104.075141}{{\em Physical Review B} {\bfseries 104} no.~7, (Aug., 2021) }. \url{http://dx.doi.org/10.1103/PhysRevB.104.075141}.

\bibitem{Gukov:2020btk}
S.~Gukov, P.-S. Hsin, and D.~Pei, ``{Generalized global symmetries of $T[M]$ theories. Part I},'' \href{http://dx.doi.org/10.1007/JHEP04(2021)232}{{\em JHEP} {\bfseries 04} (2021) 232}, \href{http://arxiv.org/abs/2010.15890}{{\ttfamily arXiv:2010.15890 [hep-th]}}.

\bibitem{Kong2020algebraic}
L.~Kong, T.~Lan, X.-G. Wen, Z.-H. Zhang, and H.~Zheng, ``Algebraic higher symmetry and categorical symmetry: A holographic and entanglement view of symmetry,'' \href{http://dx.doi.org/10.1103/physrevresearch.2.043086}{{\em Physical Review Research} {\bfseries 2} no.~4, (Oct., 2020) }. \url{http://dx.doi.org/10.1103/PhysRevResearch.2.043086}.

\bibitem{Aasen:2016dop}
D.~Aasen, R.~S.~K. Mong, and P.~Fendley, ``{Topological Defects on the Lattice I: The Ising model},'' \href{http://dx.doi.org/10.1088/1751-8113/49/35/354001}{{\em J. Phys. A} {\bfseries 49} no.~35, (2016) 354001}, \href{http://arxiv.org/abs/1601.07185}{{\ttfamily arXiv:1601.07185 [cond-mat.stat-mech]}}.

\bibitem{Chatterjee2023shadow}
A.~Chatterjee and X.-G. Wen, ``Symmetry as a shadow of topological order and a derivation of topological holographic principle,'' \href{http://dx.doi.org/10.1103/PhysRevB.107.155136}{{\em Phys. Rev. B} {\bfseries 107} (Apr, 2023) 155136}. \url{https://link.aps.org/doi/10.1103/PhysRevB.107.155136}.

\bibitem{Moradi2023holography}
H.~Moradi, S.~F. Moosavian, and A.~Tiwari, ``{Topological holography: Towards a unification of Landau and beyond-Landau physics},'' \href{http://dx.doi.org/10.21468/SciPostPhysCore.6.4.066}{{\em SciPost Phys. Core} {\bfseries 6} (2023) 066}. \url{https://scipost.org/10.21468/SciPostPhysCore.6.4.066}.

\bibitem{kaidi2023symmetrytftanomalies}
J.~Kaidi, E.~Nardoni, G.~Zafrir, and Y.~Zheng, ``Symmetry tfts and anomalies of non-invertible symmetries,'' 2023.
\newblock \url{https://arxiv.org/abs/2301.07112}.

\bibitem{bhardwaj2023charges}
L.~Bhardwaj and S.~Schafer-Nameki, ``Generalized charges, part ii: Non-invertible symmetries and the symmetry tft,'' 2023.
\newblock \url{https://arxiv.org/abs/2305.17159}.

\bibitem{Apruzzi2023symTFT}
F.~Apruzzi, F.~Bonetti, I.~García~Etxebarria, S.~S. Hosseini, and S.~Schäfer-Nameki, ``Symmetry tfts from string theory,'' \href{http://dx.doi.org/10.1007/s00220-023-04737-2}{{\em Communications in Mathematical Physics} {\bfseries 402} no.~1, (May, 2023) 895--949}. \url{http://dx.doi.org/10.1007/s00220-023-04737-2}.

\bibitem{chatterjee2023holographic}
A.~Chatterjee and X.-G. Wen, ``Holographic theory for continuous phase transitions: Emergence and symmetry protection of gaplessness,'' \href{http://dx.doi.org/10.1103/PhysRevB.108.075105}{{\em Phys. Rev. B} {\bfseries 108} (Aug, 2023) 075105}. \url{https://link.aps.org/doi/10.1103/PhysRevB.108.075105}.

\bibitem{Bhardwaj2025clubsandwich}
L.~Bhardwaj, L.~E. Bottini, D.~Pajer, and S.~Schäfer-Nameki, ``The club sandwich: Gapless phases and phase transitions with non-invertible symmetries,'' \href{http://dx.doi.org/10.21468/scipostphys.18.5.156}{{\em SciPost Physics} {\bfseries 18} no.~5, (May, 2025) }. \url{http://dx.doi.org/10.21468/SciPostPhys.18.5.156}.

\bibitem{bhardwaj2024hasse}
L.~Bhardwaj, D.~Pajer, S.~Schafer-Nameki, and A.~Warman, ``Hasse diagrams for gapless spt and ssb phases with non-invertible symmetries,'' 2024.
\newblock \url{https://arxiv.org/abs/2403.00905}.

\bibitem{bhardwaj2024gappedphases21dnoninvertible}
L.~Bhardwaj, D.~Pajer, S.~Schafer-Nameki, A.~Tiwari, A.~Warman, and J.~Wu, ``Gapped phases in (2+1)d with non-invertible symmetries: Part i,'' 2024.
\newblock \url{https://arxiv.org/abs/2408.05266}.

\bibitem{Antinucci2025gaplessSPT}
A.~Antinucci, C.~Copetti, and S.~Schäfer-Nameki, ``Symtft for (3+1)d gapless spts and obstructions to confinement,'' \href{http://dx.doi.org/10.21468/scipostphys.18.3.114}{{\em SciPost Physics} {\bfseries 18} no.~3, (Mar., 2025) }. \url{http://dx.doi.org/10.21468/SciPostPhys.18.3.114}.

\bibitem{bottini2025haagerup}
L.~E. Bottini and S.~Schafer-Nameki, ``A gapless phase with haagerup symmetry,'' 2025.
\newblock \url{https://arxiv.org/abs/2410.19040}.

\bibitem{bhardwaj2025gappedphases21dnoninvertible}
L.~Bhardwaj, S.~Schafer-Nameki, A.~Tiwari, and A.~Warman, ``Gapped phases in (2+1)d with non-invertible symmetries: Part ii,'' 2025.
\newblock \url{https://arxiv.org/abs/2502.20440}.

\bibitem{bhardwaj2025gaplessphases21dnoninvertible}
L.~Bhardwaj, Y.~Gai, S.-J. Huang, K.~Inamura, S.~Schafer-Nameki, A.~Tiwari, and A.~Warman, ``Gapless phases in (2+1)d with non-invertible symmetries,'' 2025.
\newblock \url{https://arxiv.org/abs/2503.12699}.

\bibitem{Yang_2021}
Z.-Y. Yang, J.~Yang, C.~Fang, and Z.-X. Liu, ``A hamiltonian approach for obtaining irreducible projective representations and the k.p perturbation for anti-unitary symmetry groups,'' \href{http://dx.doi.org/10.1088/1751-8121/abfffc}{{\em Journal of Physics A: Mathematical and Theoretical} {\bfseries 54} no.~26, (Jun, 2021) 265202}. \url{https://dx.doi.org/10.1088/1751-8121/abfffc}.

\bibitem{2304.01827}
K.~Shiozaki and S.~Ono, ``Atiyah-hirzebruch spectral sequence for topological insulators and superconductors: E2 pages for 1651 magnetic space groups,'' \href{http://arxiv.org/abs/2304.01827}{{\ttfamily arXiv:2304.01827 [cond-mat.mes-hall]}}.

\bibitem{PhysRevB.111.134407}
Z.~Song, A.~Z. Yang, Y.~Jiang, Z.~Fang, J.~Yang, C.~Fang, H.~Weng, and Z.-X. Liu, ``Constructions and applications of irreducible representations of spin-space groups,'' \href{http://dx.doi.org/10.1103/PhysRevB.111.134407}{{\em Phys. Rev. B} {\bfseries 111} (Apr, 2025) 134407}. \url{https://link.aps.org/doi/10.1103/PhysRevB.111.134407}.

\bibitem{Higgs1989projective}
R.~Higgs, ``Projective characters of degree one and the inflation-restriction sequence,'' \href{http://dx.doi.org/10.1017/S1446788700030731}{{\em Journal of the Australian Mathematical Society} {\bfseries 46} no.~2, (1989) 272--280}. \url{https://doi.org/10.1017/S1446788700030731}.

\bibitem{Pollmann2010spectrum}
F.~Pollmann, A.~M. Turner, E.~Berg, and M.~Oshikawa, ``Entanglement spectrum of a topological phase in one dimension,'' \href{http://dx.doi.org/10.1103/physrevb.81.064439}{{\em Physical Review B} {\bfseries 81} no.~6, (Feb., 2010) }. \url{http://dx.doi.org/10.1103/PhysRevB.81.064439}.

\bibitem{Turner2011onedimensionalphase}
A.~M. Turner, F.~Pollmann, and E.~Berg, ``Topological phases of one-dimensional fermions: An entanglement point of view,'' \href{http://dx.doi.org/10.1103/physrevb.83.075102}{{\em Physical Review B} {\bfseries 83} no.~7, (Feb., 2011) }. \url{http://dx.doi.org/10.1103/PhysRevB.83.075102}.

\bibitem{PerezGarcia2008string}
D.~Pérez-García, M.~M. Wolf, M.~Sanz, F.~Verstraete, and J.~I. Cirac, ``String order and symmetries in quantum spin lattices,'' \href{http://dx.doi.org/10.1103/physrevlett.100.167202}{{\em Physical Review Letters} {\bfseries 100} no.~16, (Apr., 2008) }. \url{http://dx.doi.org/10.1103/PhysRevLett.100.167202}.

\bibitem{Haegeman2012orderparameter}
J.~Haegeman, D.~Pérez-García, I.~Cirac, and N.~Schuch, ``Order parameter for symmetry-protected phases in one dimension,'' \href{http://dx.doi.org/10.1103/physrevlett.109.050402}{{\em Physical Review Letters} {\bfseries 109} no.~5, (July, 2012) }. \url{http://dx.doi.org/10.1103/PhysRevLett.109.050402}.

\bibitem{fechisin2023noninvertible}
C.~Fechisin, N.~Tantivasadakarn, and V.~V. Albert, ``Noninvertible symmetry-protected topological order in a group-based cluster state,'' \href{http://dx.doi.org/10.1103/physrevx.15.011058}{{\em Physical Review X} {\bfseries 15} no.~1, (2025) }, \href{http://arxiv.org/abs/2312.09272}{{\ttfamily arXiv:2312.09272 [cond-mat.str-el]}}. \url{http://dx.doi.org/10.1103/PhysRevX.15.011058}.

\bibitem{Chen2024sequential}
X.~Chen, A.~Dua, M.~Hermele, D.~T. Stephen, N.~Tantivasadakarn, R.~Vanhove, and J.-Y. Zhao, ``Sequential quantum circuits as maps between gapped phases,'' \href{http://dx.doi.org/10.1103/physrevb.109.075116}{{\em Physical Review B} {\bfseries 109} no.~7, (Feb., 2024) }. \url{http://dx.doi.org/10.1103/PhysRevB.109.075116}.

\bibitem{James_Liebeck_2001}
G.~James and M.~Liebeck, {\em Representations and Characters of Groups}.
\newblock Cambridge University Press, 2~ed., 2001.

\bibitem{seifnashri2024cluster}
S.~Seifnashri and S.-H. Shao, ``Cluster state as a non-invertible symmetry protected topological phase,'' \href{http://arxiv.org/abs/2404.01369}{{\ttfamily arXiv:2404.01369 [cond-mat.str-el]}}.

\bibitem{inamura202411dsptphasesfusion}
K.~Inamura and S.~Ohyama, ``1+1d spt phases with fusion category symmetry: interface modes and non-abelian thouless pump,'' 2024.
\newblock \url{https://arxiv.org/abs/2408.15960}.

\bibitem{Barkeshli2023codimension}
M.~Barkeshli, Y.-A. Chen, S.-J. Huang, R.~Kobayashi, N.~Tantivasadakarn, and G.~Zhu, ``Codimension-2 defects and higher symmetries in (3+1)d topological phases,'' \href{http://dx.doi.org/10.21468/scipostphys.14.4.065}{{\em SciPost Physics} {\bfseries 14} no.~4, (Apr., 2023) }. \url{http://dx.doi.org/10.21468/SciPostPhys.14.4.065}.

\bibitem{Barkeshli:2022edm}
M.~Barkeshli, Y.-A. Chen, P.-S. Hsin, and R.~Kobayashi, ``{Higher-group symmetry in finite gauge theory and stabilizer codes},'' \href{http://arxiv.org/abs/2211.11764}{{\ttfamily arXiv:2211.11764 [cond-mat.str-el]}}.

\bibitem{Barkeshli:2023bta}
M.~Barkeshli, P.-S. Hsin, and R.~Kobayashi, ``{Higher-group symmetry of (3+1)D fermionic $\mathbb{Z}_2$ gauge theory: logical CCZ, CS, and T gates from higher symmetry},'' \href{http://arxiv.org/abs/2311.05674}{{\ttfamily arXiv:2311.05674 [cond-mat.str-el]}}.

\bibitem{Hsin:2019fhf}
P.-S. Hsin and A.~Turzillo, ``{Symmetry-enriched quantum spin liquids in (3 + 1)$d$},'' \href{http://dx.doi.org/10.1007/JHEP09(2020)022}{{\em JHEP} {\bfseries 09} (2020) 022}, \href{http://arxiv.org/abs/1904.11550}{{\ttfamily arXiv:1904.11550 [cond-mat.str-el]}}.

\bibitem{Yoshida2016generalized}
B.~Yoshida, ``Topological phases with generalized global symmetries,'' \href{http://dx.doi.org/10.1103/physrevb.93.155131}{{\em Physical Review B} {\bfseries 93} no.~15, (Apr., 2016) }. \url{http://dx.doi.org/10.1103/PhysRevB.93.155131}.

\bibitem{Beigi:2010htr}
S.~Beigi, P.~W. Shor, and D.~Whalen, ``{The Quantum Double Model with Boundary: Condensations and Symmetries},'' \href{http://dx.doi.org/10.1007/s00220-011-1294-x}{{\em Commun. Math. Phys.} {\bfseries 306} no.~3, (2011) 663--694}, \href{http://arxiv.org/abs/1006.5479}{{\ttfamily arXiv:1006.5479 [quant-ph]}}.

\bibitem{Roumpedakis:2022aik}
K.~Roumpedakis, S.~Seifnashri, and S.-H. Shao, ``{Higher Gauging and Non-invertible Condensation Defects},'' \href{http://dx.doi.org/10.1007/s00220-023-04706-9}{{\em Commun. Math. Phys.} {\bfseries 401} no.~3, (2023) 3043--3107}, \href{http://arxiv.org/abs/2204.02407}{{\ttfamily arXiv:2204.02407 [hep-th]}}.

\bibitem{Kaidi2022higher}
J.~Kaidi, Z.~Komargodski, K.~Ohmori, S.~Seifnashri, and S.-H. Shao, ``Higher central charges and topological boundaries in 2+1-dimensional tqfts,'' \href{http://dx.doi.org/10.21468/scipostphys.13.3.067}{{\em SciPost Physics} {\bfseries 13} no.~3, (Sept., 2022) }. \url{http://dx.doi.org/10.21468/SciPostPhys.13.3.067}.

\bibitem{Lan2015gapped}
T.~Lan, J.~C. Wang, and X.-G. Wen, ``Gapped domain walls, gapped boundaries, and topological degeneracy,'' \href{http://dx.doi.org/10.1103/physrevlett.114.076402}{{\em Physical Review Letters} {\bfseries 114} no.~7, (Feb., 2015) }. \url{http://dx.doi.org/10.1103/PhysRevLett.114.076402}.

\end{thebibliography}\endgroup

\appendix
\section{Higgs' Example}

Here we provide another example of a finite group with nontrivial Bogomolov multiplier.
Let $p$ be a prime number. Consider the group $H$ defined by four generators $x_i$ $(i=1,2,3,4)$ with the relations:
\begin{align}
x_i^p = 1,\quad [[G,G],G] = 1.
\end{align}
Due to the second condition, all commutators $c_{ij}\coloneqq [x_i,x_j]=x_ix_jx_i^{-1}x_j^{-1}$ lie in the center $Z(H)$ of $H$. Any element $h\in H$ can be uniquely expressed as
\begin{align}
h=x_1^{n_1}x_2^{n_2}x_3^{n_3}x_4^{n_4}\prod_{1\leq j<k\leq 4}c_{jk}^{m_{jk}},
\end{align}
where $n_i, m_{jk}=0,1,\dots,p-1$ $(1\leq i\leq 4, 1\leq j<k\leq 4)$. Thus, the group order is $|H|=p^{10}$.

Consider the product of elements:
\begin{align}
h&=x_1^{n_1}x_2^{n_2}x_3^{n_3}x_4^{n_4}\prod_{1\leq j<k\leq 4}c_{jk}^{m_{jk}},\\
h'&=x_1^{n_1'}x_2^{n_2'}x_3^{n_3'}x_4^{n_4'}\prod_{1\leq j<k\leq 4}c_{jk}^{m_{jk}'}.
\end{align}
By definition of the commutator,
\begin{align}
x_j^{n_j}x_i^{n_i'}=c_{ij}^{-n_i'n_j}x_i^{n_i'}x_j^{n_j},
\end{align}
and thus
\begin{align}
hh'&=x_1^{n_1+n_1'}x_2^{n_2+n_2'}x_3^{n_3+n_3'}x_4^{n_4+n_4'}\prod_{1\leq j<k\leq 4}c_{jk}^{m_{jk}+m_{jk}'-n_j'n_k}.
\end{align}

Define $A$ as the subgroup of $Z(H)$ generated by $s\coloneqq c_{12}c_{34}$, having order $|A|=p$. Furthermore, no non-identity element $s^1,\dots,s^{p-1}$ in $A$ can be expressed as a single commutator. This $H$ is the smallest group possessing such a subgroup \cite{Higgs1989projective}.

Define the quotient group
\begin{align}
G=H/A,
\end{align}
with order $|G|=p^{9}$. Any element $g\in G$ can be written (excluding $c_{34}$) as
\begin{align}
g=x_1^{n_1}x_2^{n_2}x_3^{n_3}x_4^{n_4}\prod_{1\leq j<k\leq 4, j\leq 2}c_{jk}^{m_{jk}}.
\end{align}

For the product of $g,g'\in G$, even when $m_{34}=m_{34}'=0$, the term $c_{34}^{-n_3'n_4}$ appears. To absorb this into the product of $c_{12}$ and $s$, we define the product in $G$ as:
\begin{align}
gg'=x_1^{n_1+n_1'}x_2^{n_2+n_2'}x_3^{n_3+n_3'}x_4^{n_4+n_4'}\prod_{1\leq j<k\leq 4}c_{jk}^{m_{jk}+m_{jk}'-n_j'n_k+n_3'n_4\delta_{j1}\delta_{k2}}.
\end{align}

\subsubsection{Multiplier system}
Define a multiplier system $\alpha\in\mathbb{Z}^2(G,A)$ by
\begin{align}
\alpha(g,g')\coloneqq(e^{\frac{2\pi i}{p}})^{-n_3'n_4}.
\end{align}
This satisfies the cocycle condition:
\begin{align}
\alpha(g,g')\alpha(gg',g'')=\alpha(g,g'g'')\alpha(g',g''),
\end{align}
and all elements are $\alpha$-regular.

\subsubsection{Projective representations}
There are $p^6+p^4$ representations of dimension $p$ and $p^5-p^4-p^2$ representations of dimension $p^2$, thus:
\begin{align}
N_{\mathrm{rep}}^{(\alpha)}&=(p^6+p^4)+(p^5-p^4-p^2)=|C^{(\alpha)}(G)|.
\end{align}
The relation
\begin{align}
\sum_{A=1}^{N_{\mathrm{rep}}^{(\alpha)}}d_A^2=p^9=|G|
\end{align}
holds.

\subsubsection{Linear representations}
There are $p^4$ one-dimensional representations, $(p^4-1)p^2$ $p$-dimensional representations, and $p^5-p^4$ $p^2$-dimensional representations, thus:
\begin{align}
N_{\mathrm{rep}}^{(0)}=p^6+p^5-p^2=|C^{(\alpha)}(G)|,
\end{align}
and
\begin{align}
\sum_{A=1}^{N_{\mathrm{rep}}^{(\alpha)}}(d_A^{(0)})^2=p^9=|G|.
\end{align}

\section{Conjugacy classes of Pollmann-Turner group}
\label{app:conjugacyclasses}
Let us consider the group presented in Sec.~\ref{sec:pollmann-turner}.
There are 32 distinct conjugacy classes in the group:
\begin{align}
\begin{split}
    &[1],[a],[b],[c_1],[c_2],[x],[y_1],[y_2], [ab], [ac_1], [ac_2], [bc_1], [bc_2], [by_1], [by_2], [c_1c_2], [c_1y_2], [c_2y_1],[xy_1],[xy_2], \\
    & [y_1y_2],[abc_1],[abc_2],[ac_1c_2], [bc_1c_2],[bc_1y_2],[bc_2y_1], [by_1y_2],[c_2c_2y_1], [xy_1y_2],[abc_1c_2],[bc_1c_2y_1]~.
    \end{split}
\end{align}
Among these 32 conjugacy classes, there are 8 elements with $n_a(C)= 1$,
\begin{align}
    n_a(C)=1: \quad C=[a],[ab], [ac_1], [ac_2],[abc_1],[abc_2],[ac_1c_2],[abc_1c_2]~.
\end{align}
There are 14 elements with $n_b(C)=1$,
\begin{align}
\begin{split}
    &n_b(C)=1: \ [b],[ab], [bc_1], [bc_2], [by_1], [by_2], [abc_1],[abc_2],[bc_1c_2],[bc_1y_2],[bc_2y_1], [by_1y_2],[abc_1c_2],[bc_1c_2y_1]~.
    \end{split}
\end{align}
There are 14 elements with $n_a(C)+n_b(C)=1$ mod 2,
\begin{align}
&n_a(C)+n_b(C)=1: \ [a],[ac_1], [ac_2],[ac_1c_2],[b], [bc_1], [bc_2], [by_1], [by_2],[bc_1c_2],[bc_1y_2],[bc_2y_1], [by_1y_2],[bc_1c_2y_1]~.
\end{align}
Using these counting one can perform \eqref{eq:computation of dpi}.

\end{document}